\documentclass{article} 
\usepackage[final,nonatbib]{nips_2017}

\usepackage{hyperref}
\usepackage{url}

\usepackage{tikz}
\usepackage{amsmath}
\usepackage{pgfplots}
\usepackage{pgfplotstable}
\usepackage{booktabs}
\usepackage{notes}
\usepackage{url}
\usepackage[ruled, vlined, linesnumbered, nofillcomment]{algorithm2e}
\usepackage{multirow}
\usepackage{caption}
\usepackage{subfig}
\usepackage{xfrac}
\usepackage{fix-cm} 
\usepackage{amsthm}
\usepackage{amssymb}
\usepackage[numbers]{natbib}

\newcommand{\spara}[1]{\smallskip\noindent\textbf{#1}}

\newcommand{\para}[1]{\noindent\textbf{#1}}

\newtheorem{definition}{Definition}
\newtheorem{lemma}[definition]{Lemma}
\newtheorem{proposition}[definition]{Proposition}
\newtheorem{problem}{Problem}[section]
\newtheorem{subproblem}[definition]{Subproblem}

\newcommand{\set}[1]{\left\{#1\right\}}
\newcommand{\pr}[1]{\left(#1\right)}
\newcommand{\fpr}[1]{\mathopen{}\left(#1\right)}

\newcommand{\fspr}[1]{\mathopen{}\left[#1\right]}
\newcommand{\brak}[1]{\left<#1\right>}
\newcommand{\abs}[1]{{\left|#1\right|}}

\newcommand{\real}{\mathbb{R}}
\newcommand{\np}{\textbf{NP}}
\newcommand{\countingp}{\textbf{\#P}}

\newcommand{\funcdef}[3]{{#1}:{#2} \to {#3}}
\newcommand{\define}{\leftarrow}

\newcommand{\NP}{\ensuremath{\mathbf{NP}}}
\newcommand{\xor}{\mathbin{\triangle}}

\DeclareRobustCommand{\dispfunc}[2]{%
	\ensuremath{%
		\ifthenelse{\equal{#2}{}}%
		{\mathit{#1}}%
		{\mathit{#1}\fpr{#2}}}}

\newcommand{\pone}{p_1}
\newcommand{\ptwo}{p_2}
\newcommand{\pii}{p_i}
\newcommand{\reach}{r}
\newcommand{\reachone}{{\reach}_1}
\newcommand{\reachtwo}{{\reach}_2}
\newcommand{\diffcover}{\Omega}

\newcommand{\diffnew}{\Psi}
\newcommand{\diff}{\Phi}
\newcommand{\diffind}{{\diff}_{\mathit{H}}}
\newcommand{\difftc}{{\diff}_{\mathit{C}}}
\newcommand{\sideone}{q_1}
\newcommand{\sidetwo}{q_2}
\newcommand{\sideii}{q_i}
\newcommand{\retweets}{R}

\newcommand{\mean}[2]{\operatorname{E}_{#1}\fspr{#2}}

\newcommand{\meantc}[1]{\mean{\mathit{C}}{#1}}

\newcommand{\balanceprob}{{\sc Balance}\xspace}
\newcommand{\balanceprobh}{{\sc Balance-H}\xspace}
\newcommand{\balanceprobc}{{\sc Balance-C}\xspace}

\newcommand{\greedyalg}{\textsf{\small Greedy}\xspace}
\newcommand{\coveralg}{\textsf{\small Cover}\xspace}
\newcommand{\commonalg}{\textsf{\small Common}\xspace}
\newcommand{\hedgealg}{\textsf{\small Hedge}\xspace}
\newcommand{\randomalg}{\textsf{\small Random}\xspace}
\newcommand{\highdegreealg}{\textsf{\small High\-Degree}\xspace}
\newcommand{\interalg}{\textsf{\small Intersection}\xspace}
\newcommand{\borodinalg}{\textsf{\small BBLO}\xspace}
\newcommand{\unionalg}{\textsf{\small Union}\xspace}

\newcommand{\uselections}{{\tt US-elections}\xspace}
\newcommand{\uselectionsalt}{{\tt US}\xspace}
\newcommand{\brexit}{{\tt Brexit}\xspace}
\newcommand{\iphone}{{\tt iPhone}\xspace}
\newcommand{\obamacare}{{\tt ObamaCare}\xspace}
\newcommand{\abortion}{{\tt Abortion}\xspace}
\newcommand{\fracking}{{\tt Fracking}\xspace}

\newcommand{\squishlist}{
	\begin{list}{$\bullet$}
		{  \setlength{\itemsep}{0pt}
			\setlength{\parsep}{3pt}
			\setlength{\topsep}{3pt}
			\setlength{\partopsep}{0pt}
			\setlength{\leftmargin}{1em}
			\setlength{\labelwidth}{1.5em}
			\setlength{\labelsep}{0.5em}
	} }
	\newcommand{\squishlisttight}{
		\begin{list}{$\bullet$}
			{ \setlength{\itemsep}{0pt}
				\setlength{\parsep}{0pt}
				\setlength{\topsep}{0pt}
				\setlength{\partopsep}{0pt}
				\setlength{\leftmargin}{2em}
				\setlength{\labelwidth}{1.5em}
				\setlength{\labelsep}{0.5em}
		} }
		\newcommand{\squishdesc}{
			\begin{list}{}
				{  \setlength{\itemsep}{0pt}
					\setlength{\parsep}{3pt}
					\setlength{\topsep}{3pt}
					\setlength{\partopsep}{0pt}
					\setlength{\leftmargin}{1em}
					\setlength{\labelwidth}{1.5em}
					\setlength{\labelsep}{0.5em}
			} }
			\newcommand{\squishend}{
			\end{list}
		}
		
		\definecolor{yafaxiscolor}{rgb}{0.3, 0.3, 0.3}
		
		\definecolor{yafcolor1}{rgb}{0.4, 0.165, 0.553}
		\definecolor{yafcolor2}{rgb}{0.949, 0.482, 0.216}
		\definecolor{yafcolor3}{rgb}{0.47, 0.549, 0.306}
		\definecolor{yafcolor4}{rgb}{0.925, 0.165, 0.224}
		\definecolor{yafcolor5}{rgb}{0.141, 0.345, 0.643}
		\definecolor{yafcolor6}{rgb}{0.965, 0.933, 0.267}
		\definecolor{yafcolor7}{rgb}{0.627, 0.118, 0.165}
		\definecolor{yafcolor8}{rgb}{0.878, 0.475, 0.686}
		
		\tikzstyle{exnode} = [inner sep = 1pt]
		\tikzstyle{labnode} = [sloped, text = black, font = \scriptsize, inner sep = 1pt]
		\tikzstyle{exedge} = [yafcolor5, draw, thick, >=latex, ->]
		\tikzstyle{exedge2} = [yafcolor2, draw, thick, >=latex, ->]
		\tikzstyle{exedge3} = [yafcolor3, draw, thick, >=latex, ->]

		\newlength{\yafaxispad}
		\newlength{\yaftlpad}
		\newlength{\yaflabelpad}
		\newlength{\yafaxiswidth}
		\newlength{\yafticklen}
		\makeatletter
		\def\pgfplots@drawtickgridlines@INSTALLCLIP@onorientedsurf#1{}
		\makeatother
		
		\newcommand{\yafdrawxaxis}[2]{
			\pgfplotstransformcoordinatex{#1}\let\xmincoord=\pgfmathresult 
			\pgfplotstransformcoordinatex{#2}\let\xmaxcoord=\pgfmathresult 
			\pgfsetlinewidth{\yafaxiswidth} 
			\pgfsetcolor{yafaxiscolor}
			\pgfpathmoveto{\pgfpointadd{\pgfpointadd{\pgfplotspointrelaxisxy{0}{0}}{\pgfqpointxy{\xmincoord}{0}}}{\pgfqpoint{-0.5\yafaxiswidth}{\yafaxispad}}}
			\pgfpathlineto{\pgfpointadd{\pgfpointadd{\pgfplotspointrelaxisxy{0}{0}}{\pgfqpointxy{\xmaxcoord}{0}}}{\pgfqpoint{0.5\yafaxiswidth}{\yafaxispad}}}
			\pgfusepath{stroke}
			
		}
		\newcommand{\yafdrawyaxis}[2]{
			\pgfplotstransformcoordinatey{#1}\let\ymincoord=\pgfmathresult 
			\pgfplotstransformcoordinatey{#2}\let\ymaxcoord=\pgfmathresult 
			\pgfsetlinewidth{\yafaxiswidth} 
			\pgfsetcolor{yafaxiscolor}
			\pgfpathmoveto{\pgfpointadd{\pgfpointadd{\pgfplotspointrelaxisxy{0}{0}}{\pgfqpointxy{0}{\ymincoord}}}{\pgfqpoint{\yafaxispad}{-0.5\yafaxiswidth}}}
			\pgfpathlineto{\pgfpointadd{\pgfpointadd{\pgfplotspointrelaxisxy{0}{0}}{\pgfqpointxy{0}{\ymaxcoord}}}{\pgfqpoint{\yafaxispad}{0.5\yafaxiswidth}}}
			\pgfusepath{stroke}
		}
		
		\newcommand{\yafdrawaxis}[4]{\yafdrawxaxis{#1}{#2}\yafdrawyaxis{#3}{#4}}
		
		\pgfplotscreateplotcyclelist{yaf}{%
			{yafcolor1,mark options={scale=0.75},mark=*}, 
			{yafcolor2,mark options={scale=0.75},mark=square*},
			{yafcolor3,mark options={scale=1.3},mark=x},
			{yafcolor5,mark options={scale=0.75},mark=triangle*},
			{yafcolor4,mark options={scale=0.75},mark=o},
			{yafcolor6,mark options={scale=0.75},mark=o},
			{yafcolor7,mark options={scale=0.75},mark=o},
			{yafcolor8,mark options={scale=0.75},mark=o}} 
		
		\pgfplotsset{axis y line=left, axis x line=bottom,
			tick align=outside,
			compat = 1.3,
			tickwidth=\yafticklen,
			clip = false,
			every axis title shift = 0pt,
			x axis line style= {-, line width = 0pt, opacity = 0},
			y axis line style= {-, line width = 0pt, opacity = 0},
			x tick style= {line width = \yafaxiswidth, color=yafaxiscolor, yshift = \yafaxispad},
			y tick style= {line width = \yafaxiswidth, color=yafaxiscolor, xshift = \yafaxispad},
			x tick label style = {font=\scriptsize, yshift = \yaftlpad},
			y tick label style = {font=\scriptsize, xshift = \yaftlpad},
			every axis y label/.style = {at = {(ticklabel cs:0.5)}, rotate=90, anchor=center, font=\scriptsize, yshift = -\yaflabelpad},
			every axis x label/.style = {at = {(ticklabel cs:0.5)}, anchor=center, font=\scriptsize, yshift = \yaflabelpad},
			x tick label style = {font=\scriptsize, yshift = 1pt},
			grid = major,
			major grid style  = {dash pattern = on 1pt off 3 pt},
			every axis plot post/.append style= {line width=\yafaxiswidth} ,
			legend cell align = left,
			legend style = {inner ysep = 1pt, inner xsep = 3pt, cells = {font=\scriptsize}},
			legend image code/.code={%
				\draw[mark repeat=2,mark phase=2,#1] 
				plot coordinates { (0cm,0cm) (0.15cm,0cm) (0.3cm,0cm) };%
			} 
		}

\author{
	Kiran Garimella \\
	Aalto University \& HIIT\\
	Helsinki, Finland\\
	\texttt{kiran.garimella@aalto.fi}
	\And
	Aristides Gionis\\
	Aalto University \& HIIT\\
	Helsinki, Finland\\
	\texttt{aristides.gionis@aalto.fi}
	\And
	Nikos Parotsidis\\
	University of Rome Tor Vergata\\
	Rome, Italy\\
	\texttt{nikos.parotsidis@uniroma2.it}
	\And 
	Nikolaj Tatti \\
	Aalto University \& HIIT\\
	Helsinki, Finland\\
	\texttt{nikolaj.tatti@aalto.fi}
}

\begin{document}
\title{Balancing information exposure in social networks}

\maketitle

\begin{abstract}
Social media has brought a revolution 
on how people 
are consuming news.
Beyond the undoubtedly large number of 
advantages 
brought by social-media platforms, 
a point of criticism
has been the creation of \emph{echo chambers} and \emph{filter bubbles},
caused by \emph{social homophily}
and \emph{algorithmic personalization}.

In this paper
we address the problem of 
\emph{balancing the information exposure}
in a social network.
We assume that two opposing campaigns (or viewpoints) 
are present in the network,
and that network nodes have different preferences 
towards these campaigns.
Our goal is to find two sets of nodes
to employ in the respective campaigns, so that 
the overall information exposure for the two campaigns is {\em balanced}.
We formally define the problem, 
characterize its hardness, 
develop approximation algorithms, 
and present experimental evaluation results.

Our model is inspired by the literature on \emph{influence maximization}, 
but we offer significant novelties. 
First, {\em balance} of information exposure is modeled by a
{\em symmetric difference} function, 
which is neither monotone nor submodular, 
and thus, not amenable to existing approaches. 
Second, while previous papers consider 
a setting with selfish agents
and provide bounds on {\em best response} strategies
(i.e., move of the last player), 
we consider a setting with a centralized agent and
provide bounds for a global objective function.

\end{abstract}

%
%

%
%



\section{Introduction}
\label{sec:intro}
Social-media platforms have revolutionized many aspects of human culture, 
among others, the way people are exposed to information.
A recent survey estimates that 62\% of adults in the US
get their news on social media~\cite{pew}.
Despite providing many desirable features, such as, 
searching, personalization, and recommendations,
one point of criticism is that 
social media amplify the phenomenon of \emph{echo chambers} and \emph{filter bubbles}:
users get less exposure to conflicting viewpoints and are isolated  
in their own informational bubble.
This phenomenon is contributed to social homophily 
and algorithmic personalization, 
and is more acute for controversial topics~\cite{akoglu2014quantifying,conover2011political,del2015echo,garimella2016quantifying,garrett2009echo}.

In this paper we address the problem
of reducing the filter-bubble effect 
by balancing information exposure among users. 
We consider social-media discussions around a topic
that are characterized by two or more {\em conflicting viewpoints}.
Let us refer to these viewpoints as {\em campaigns}.
Our approach follows the popular paradigm of influence propagation~\cite{kempe2003maximizing}:
we want to select a small number of seed users for each campaign
so as to maximize the 
number of users who are {\em exposed to both campaigns}.
In contrast to existing work on competitive viral marketing, 
we do not consider the problem of finding an optimal {\em selfish strategy} 
for each campaign separately.  
Instead we consider a {\em certalized agent}
responsible for balancing information exposure for the two campaings. 
Consider the following motivating examples.

\noindent
{\bf Example 1:}
Prominent social-media companies, like Facebook and Twitter,
have been called to act as arbiters so as to prevent ideological isolation and 
polarization in the society. 
The motivation for companies to assume this role
could be for improving their public image or due to government policies.\footnote{For instance, 
Germany is now fining Facebook for the spread of fake news.}
Consider a controversial topic being discussed in social-media platform $X$, 
which has led to polarization and filter bubbles.
Platform $X$ has the ability to algorithmically detect
such bubbles~\cite{garimella2016quantifying}, 
identify the influential users on each side, 
and estimate the influence among users~\cite{du2013scalable,gomez2010inferring}.
As part of a new filter-bubble bursting service, 
platform $X$ would like to disseminate two high-quality and thought-provoking dueling op-eds,
articles, one for each side, that present the arguments of the other side in a fair manner.
Assume that $X$ is interested in following a viral-marketing approach.
Which users should $X$ target, for each of the two articles,
so that people in the network are informed in the most balanced way?

\noindent
{\bf Example 2:}
Government organization $Y$ is initiating a programme to help
assimilate foreigners who have newly arrived in the country. 
Part of the initiative focuses on bringing the communities of 
foreigners and locals closer in social media. 
Organization $Y$ is interested in identifying individuals who can help spreading news
of one community into the other.

From the technical standpoint, 
we consider the following problem setting:
We assume that information is propagated 
in the network according to the {\em independent-cascade model}~\cite{kempe2003maximizing}.
We assume that there are two opposing campaigns,
and for each one there is a set of initial seed nodes, $I_1$ and $I_2$, 
which are not necessarily distinct. 
Furthermore, 
we assume that the users in the network are exposed to information 
about campaign $i$ via diffusion from the set of seed nodes $I_i$.
The diffusion in the network may occur with 
independent or correlated probabilities
for the two campaigns; 
we consider both settings to which we are referring as 
{\em heterogeneous} or {\em correlated}.

The objective is to recruit two additional sets of seed nodes, 
$S_1$ and $S_2$, for the two campaigns, 
with $\abs{S_1} + \abs{S_2} \leq k$, for a given budget $k$, 
so as to maximize the expected number of balanced users,
i.e., the users who are exposed to 
information from both campaigns (or from none!).

We formally define the problem of 
balancing information exposure and we show that it is \NP-hard.
We develop different approximation algorithms
for the different settings we consider, 
as well as heuristic variants of the proposed algorithm. 
We experimentally evaluate our methods, 
on several real-world (and realistic) datasets,
collected from twitter,
for different topics of interest.

Although our approach is inspired by the large body of work on 
information propagation, 
and resembles previous problem formulations 
for competitive viral marketing, 
there are significant differences and novelties. 
In particular:
\squishlist
\item 
This is the first paper to address the 
problem of {\em balancing information exposure}
and {\em breaking filter bubbles},
using the information-propagation methodology.
\item 
The objective function that best suits our problem setting
is related to the {\em size of the symmetric difference}
of users exposed to the two campaigns. 
This is in contrast to previous settings
that consider functions related to the {\em size of the coverage} 
of the campaigns. 
\item 
As a technical consequence of the previous point, 
our objective function is neither {\em monotone} nor {\em submodular} 
making our problem more challenging.
Yet we are able to analyze the problem structure
and provide algorithms with approximation guarantees.
\item 
While most previous papers consider 
selfish agents,
and provide bounds on {\em best-response} strategies
(i.e., move of the last player), 
we consider a centralized setting and
provide bounds for a global objective function.
\squishend
We note that our datasets and implementation is 
publicly available.\footnote{\url{https://users.ics.aalto.fi/kiran/BalanceExposure/}}

\section{Related Work}
\label{sec:relatedwork}
\spara{Detecting and breaking filter bubbles.}
Several studies have observed
that users in online social networks  
prefer to associate with like-minded individuals
and consume agreeable content. 
This phenomenon 
leads to {\em filter bubbles}, 
{\em echo chambers}~\cite{pariser2011filter,nikolov2015measuring}, 
and to online polarization~\cite{adamic2005political,akoglu2014quantifying,beutel2012interacting,conover2011political,garimella2016quantifying,guerra2013measure,morales2015measuring}.
Once these filter bubbles are detected, 
the next step is to try to overcome them. 
One way to achieve this is by making 
recommendations to individuals of opposing viewpoints. 
This idea has been explored, in different ways,
by a number of studies in the literature~\cite{garimella2017reducing,liao2014can,liao2014expert,munson2013encouraging,vydiswaran2015overcoming}. 
However, all previous studies
address the problem of breaking filter bubbles
by the means of {\em content recommendation}.
To the best of our knowledge, 
this is the first paper that considers an
{\em information diffusion} approach.

\spara{Information diffusion.}
Following a large body of work,
we model diffusion using the {\em independent-cascade model}~\cite{kempe2003maximizing}.
%
In the basic 
model a single item
propagates in the network.
An extension 
is when multiple items propagate simultaneously.
All works that study optimization problems
in the case of multiple items,
consider that items {\em compete} for being adopted by users.
In other words, every user adopts at most one of the existing items
and participates in at most one cascade. 

Myers and Leskovec~\cite{myers2012clash} argue that spreading processes 
may either cooperate or compete.
Competing contagions decrease each other's probability of diffusion, 
while cooperating ones help each other in being adopted. 
They propose a model that quantifies how different spreading cascades 
interact with each other. 
Carnes et al.~\cite{carnes2007maximizing} propose two models
for competitive diffusion.
Subsequently, several other 
models have been proposed~\cite{apt2011diffusion, borodin2010threshold,broecheler2010scalable,dubey2006competing,jie2016study,kostka2008word,lu2015competition,valera2015modeling}.

Most of the work on {\em competitive information diffusion}
consider the problem of selecting the best $k$ seeds for one campaign, 
for a given objective, 
in the presence of 
competing campaigns~\cite{bharathi2007competitive,budak2011limiting,nguyen2012containment}. 
Bharathi et al.~\cite{bharathi2007competitive} show that, 
if all campaigns but one have fixed sets of seeds, 
the problem for selecting the seeds for the last player is submodular, 
and thus, obtain an approximation algorithm 
for the strategy of the last player.
Game theoretic aspects of competitive cascades in social networks, 
including the investigation of conditions for the existence of Nash equilibrium, 
have also been studied~\cite{alon2010note,goyal2014competitive,tzoumas2012game}.

The work that is most related to ours, 
in the sense of considering a {\em centralized authority},
is the one by Borodin et al.~\cite{borodin2017strategyproof}. 
They study the problem where multiple campaigns 
wish to maximize their influence by selecting a set of seeds with bounded cardinality. 
They propose a centralized mechanism to allocate sets of seeds 
(possibly overlapping) to the campaigns 
so as to maximize the social welfare, 
defined as the sum of the individual's selfish objective functions. 
One can choose any objective functions as long as it is
submodular and non-decreasing.
Under this assumption they provide strategyproof (truthful) algorithms that offer
guarantees on the social welfare.
Their framework applies for several competitive influence models.
In our case, the number of balanced users is not submodular, and so
we do not have any approximation guarantees.
Nevertheless, we can use this framework as a heuristic baseline, 
which we do in the experimental section.


\section{Problem Definition}
\label{sec:problem}

\spara{Preliminaries:}
We start with a directed graph $G=(V, E,$  $\pone, \ptwo)$
representing a social network. 
We assume that there are two distinct
campaigns that propagate through the network.
Each edge $e = (u,v)\in E$ 
is assigned two probabilities, 
$\pone(e)$ and $\ptwo(e)$,  
representing the probability
that a post from vertex $u$ will propagate 
(e.g., it will be reposted) 
to vertex $v$ in the respective campaigns.

\spara{Cascade model:}
We assume that information on the two 
campaigns propagates in the network
following the {in\-de\-pen\-dent-cascade model}~\cite{kempe2003maximizing}.
For instance, consider the propagation of the first campaign.  
The procedure for the second campaign is analogous.  
We assume that there exists a set of seeds $I_1$ 
from which the propagation process begins. 
These are
vertices in the network that support the campaign and are active in generating 
content in favor of the campaign.
Propagation in the independent-cascade model
proceeds in rounds.
At each round, there exists a set of active vertices $A_1$ 
(initially, $A_1 = I_1$), where each vertex $u \in A_1$ attempts to
activate each vertex $v \notin A_1$, such that $(u,v) \in E$, 
with probability $\pone(u,v)$.  
If the propagation attempt from a vertex $u$ to a vertex
$v$ is successful, we say that $v$ propagates the first campaign.  
At the end of each round, $A_1$ is set to be the set of vertices that propagated
the campaign during the current round.

Given a seed set $S$, we write $\reachone(S)$ and $\reachtwo(S)$ for the
vertices that are reached from $S$ using the aforementioned cascade process, for the respective campaign.
Note that since this process is random, both  
$\reachone(S)$ and $\reachtwo(S)$ are random variables.
Computing the expected number of active vertices is a \countingp-hard problem,  
however, we can approximate it within an arbitrary small factor $\epsilon$, 
with high probability, via Monte-Carlo simulations.
Due to this obstacle, all approximation algorithms that evaluate an objective function over diffusion processes immediately reduce their approximation by an additive~$\epsilon$.
Throughout this work we avoid repeating this fact for the sake of simplicity of the notation.

\spara{Heterogeneous vs.\ correlated propagations:}
Our model needs also to 
specify how the propagation on the two campaigns
interact with each other. 
We consider two settings: 
In the first setting, we assume that the campaign messages propagate
independently of each other. 
Given an edge $e = (u, v)$, 
the vertex $v$ is activated on the first campaign 
with probability $\pone(e)$, 
given that vertex $u$ is activated on the first campaign. 
Similarly, $v$ is activated on the second campaign 
with probability $\ptwo(e)$, 
given that $u$ is activated on the second campaign. 
We refer to this setting
as {\em heterogeneous}.\footnote{Although 
{\em independent} is probably a better term than {\em heterogeneous}, 
we adopt the latter to avoid any confusion with the 
independent-cascade model.}
In the second setting we assume that $\pone(e) = \ptwo(e)$, 
for each edge $e$. 
We further assume that the {\em coin flips}
for the propagation of the two campaigns 
are totally correlated. 
Namely, consider an edge $e = (u, v)$, 
where $u$ is reached by either or both campaigns.
Then with probability $\pone(e)$, \emph{any} campaign that has reached $u$, 
will also reach $v$.
We refer to this second setting
as {\em correlated}.

Note that in both settings, 
a vertex may be active by
{\em none}, {\em either}, or {\em both} campaigns.
This is in contrast to most existing work 
in competitive viral marketing, 
where it is assumed that a vertex can be activated by 
{\em at most one} campaign.
The intuition is that in our setting
activation means merely passing a message or posting an article, 
and it does not imply full commitment to the campaign.
We also note that the heterogeneous
setting is more {\em realistic} than the correlated, 
however, we also study the correlated model
as it is mathematically simpler.

\spara{Problem definition:}
We are now ready to state our problem
for {\em balancing information exposure} (\balanceprob).
Given a directed graph, initial seed sets for both campaigns
and a budget, we ask to find additional seeds that would balance the vertices.
More formally:

\begin{problem}[{\balanceprob}]
\label{problem:balance}
Let\ $G=(V,E, \pone, \ptwo)$ be a directed graph, and
two sets $I_1$ and $I_2$ of initial seeds of the two campaigns.
Assume that we are given a budget $k$.
Find two sets $S_1$ and $S_2$, 
where $\abs{S_1} + \abs{S_2} \leq k$ maximizing
\[
	\diff(S_1, S_2) = \mean{}{\abs{V \setminus \pr{\reachone(I_1 \cup S_1) \xor \reachtwo(I_2 \cup S_2)}}}.
\]
\end{problem}
The objective function $\diff(S_1, S_2)$ 
is the expected number of vertices that are either reached by both campaigns
or remain oblivious to both campaigns.
Problem~\ref{problem:balance} is defined for both settings, {\em heterogeneous}
and {\em correlated}.  When we need to make explicit the underlying setting we
refer to the respective problems by {\balanceprobh}
and {\balanceprobc}.
When referring to \balanceprobh, we denote the objective by $\diffind$.
Similarly,
when referring to \balanceprobc, we write $\difftc$.
We drop the indices, when we are referring to both models simultaneously.

\spara{Computational complexity:}
As expected, the optimization problem {\balanceprob} 
turns out to be \np-hard for both settings, heterogeneous and correlated.  
A straightforward
way to prove it is by setting $I_2 = V$, so the problems reduce to standard
influence maximization. However, we provide a stronger result.  Note
that instead of maximizing balanced vertices we can equivalently minimize the
imbalanced vertices. However, this turns to be a more difficult problem.

\begin{proposition}
\label{prop:nphard}
Assume a graph $G=(V,E, \pone, \ptwo)$ with two sets $I_1$ and $I_2$ and a budget $k$.
It is an \np-hard problem to decide whether there are sets $S_1$ and $S_2$
such that $\abs{S_1} + \abs{S_2} \leq k$ and
$
	\mean{}{\abs{\reachone(I_1 \cup S_1) \xor \reachtwo(I_2 \cup S_2)}} = 0.
$
\end{proposition}
\begin{proof}
	To prove the hardness we will use \textsc{set cover}.
	Here, we are given a universe $U$ and family of sets $C_1, \ldots, C_\ell$, and we are asked
	to select $k$ sets covering the universe $U$.
	
	To map this instance to our problem, we first define vertex set $V$ to consist
	of 3 parts, $V_1$, $V_2$ and $V_3$.  The first part corresponds to the universe $U$. The second part
	consists of $k$ copies of $\ell$ vertices, $i$th vertex in $j$th copy
	corresponds to $C_i$. 
	The third part consists of $k$ vertices $b_j$. The edges
	are as follows: a vertex $v$ in the $j$th copy, corresponding to a set $C_i$ is
	connected to the vertices corresponding to the elements in $C_i$, furthermore
	$v$ is connected to $b_j$. We set $\pone = \ptwo = 1$.
	The initial seeds are $I_1 = \emptyset$ and $I_2 = V_1 \cup V_3$.
	We set the budget to $2k$.
	
	Assume that there is a $k$-cover, $C_{i_1},\ldots, C_{i_k}$.
	We set 
	\[
	S_1 = S_2 = \set{\text{vertex corresponding to }  C_{i_j} \text{ in $j$th copy }}.
	\]
	It is easy to see that the imbalanced vertices in $I_2$ are exposed to the first campaign.
	Moreover, $S_1$ and $S_2$ do not introduce new imbalanced vertices. This makes the objective
	equals to 0.
	
	Assume that there exists a solution $S_1$ and $S_2$ with a zero cost.
	We claim that $\abs{S_1 \cap (V_1 \cup V_2)} \leq k$.
	To prove this, first note that $S_1 \cap V_2 = S_2 \cap V_2$,
	as otherwise vertices in $V_2$ are left unbalanced.
	Let $m = \abs{S_1 \cap V_2}$. Since $V_3$ must be balanced and each vertex in $V_2$ has only one edge to a vertex in $V_3$, there at least $ k $ vertices in $\abs{S_1 \cap \{V_2 \cup V_3\}}$, that is, we
	must have $\abs{S_1 \cap V_3} \geq k - m$.
	Let us write $d_{ij} = \abs{S_i \cap V_j}$.
	The budget constraints guarantee that
	\[
	d_{11} + d_{12} + d_{22} + d_{13} \leq \sum_{ij} d_{ij} \leq 2k,
	\]
	which can be rewritten as
	\[
	d_{11} + d_{12} \leq 2k - d_{22} - d_{13} \leq 2k - m - (k - m) = k.
	\]
	Construct $C$ as follows: for each $S_1 \cap V_2$, select the set that
	correponds to the vertex, for each $S_1 \cap V_1$, select any set that contain
	this vertex (there is always at least one set, otherwise the problem is
	trivially false).
	Since $V_1$ must be balanced, $C$ is a $k$-cover of $U$.
\hfill\end{proof}

This result holds for both models, 
even when $\pone = \ptwo = 1$.
This result implies that the minimization version of the problem is
\np-hard, and there is no algorithm with multiplicative approximation
guarantee. It also implies that
\balanceprobh and \balanceprobc are also \np-hard.
However, we will see later that we can obtain approximation
guarantees for these maximization problems.

\section{Greedy algorithms yielding approximation guarantees}
\label{sec:algorithm}
In this section we propose three greedy algorithms.
The first algorithm yields an approximation guarantee of $(1 - 1/e)/2$ for both
models. The remaining two algorithms yield a guarantee for the correlated
model only.

\spara{Decomposing the objective:}
Recall that
the objective function of the {\balanceprob} problem is  $\diff(S_1, S_2)$.
In order to show that this function admits an approximation guarantee, 
we decompose it into two components.
To do that, assume that we are given initial seeds $I_1$ and $I_2$, and let
us write
$
	X = \reachone(I_1) \cup \reachtwo(I_2), Y = V \setminus X.
$ 
Here $X$ are vertices reached by any initial seed 
in the two campaigns and $Y$ are the
vertices that are not reached at all. 
Note that $X$ and $Y$ are random variables.
Since $X$ and~$Y$ partition $V$, we can decompose the score $\diff(S_1, S_2)$ 
as
\[
\begin{split}
	\diff(S_1, S_2) & = \diffcover(S_1, S_2) + \diffnew(S_1, S_2), \quad\text{where} \\
	\diffcover(S_1, S_2) & = \mean{}{\abs{X \setminus \pr{\reachone(I_1 \cup S_1) \xor \reachtwo(I_2 \cup S_2)}}}, \\
	\diffnew(S_1, S_2) & = \mean{}{\abs{Y \setminus \pr{\reachone(I_1 \cup S_1) \xor \reachtwo(I_2 \cup S_2)}}}.
\end{split}
\]
We first show that $\diffcover(S_1, S_2)$ is monotone and submodular. 
It is well-known that for maximizing a function 
that has these two properties under a size constraint, 
the greedy algorithm computes an $(1-\frac{1}{e})$ 
approximate solution~\cite{nemhauser1978analysis}.

\begin{lemma}
	\label{lemma:static-joint-common-budget}
	$\diffcover(S_1, S_2)$ is monotone and submodular.
\end{lemma}
Before providing the proof, 
as a technicality, 
note that submodularity is usually
defined for functions with one argument. 
Namely, given a universe of items $U$, 
we consider functions of the type
$f:2^U\rightarrow\mathbb{R}$.
However, 
by taking $U=V\times\{1,2\}$
we can equivalently write our objectives 
as functions with one argument, i.e., 
$\diff, \diffcover, \diffnew:2^U\rightarrow\mathbb{R}$.


\begin{proof}
	The objective counts 3 types of vertices:
	(\emph{i})
	vertices covered by both initial seeds,
	(\emph{ii})
	additional vertices covered by $I_1$ and $S_2$, and
	(\emph{iii})
	additional vertices covered by $I_2$ and $S_1$. This allows us to decompose the objective as
	\[
	\diffcover(S_1, S_2) = \mean{}{\abs{A} + \abs{B} + \abs{C}},\quad\text{where}
	\]
	\[
	A  = \reachone(I_1) \cap \reachtwo(I_2), \quad
	B  = (\reachone(I_1) \setminus \reachtwo(I_2)) \cap \reachtwo(S_2), \quad
	C  = (\reachtwo(I_2) \setminus \reachone(I_1)) \cap \reachone(S_1). 
	\]
	
	Note that $A$ does not depend on $S_1$ and $S_2$.
	$B$ grows in size as we add more vertices to $S_2$, and
	$C$ grows in size as we add more vertices to $S_1$.
	This proves that the objective is monotone.
	
	To prove the submodularity,
	let us introduce some notation:
	given a set of edges $F$, we write $\reach(S; F)$
	to be the set of vertices that can be reached from $S$ via $F$.
	This allows us to define
	\[
	\begin{split}
	A(F_1, F_2) & = \reach(I_1; F_1) \cap \reach(I_2; F_2), \\
	B(F_1, F_2) & = (\reach(I_1; F_1) \setminus \reach(I_2; F_2)) \cap \reach(S_2; F_2), \\
	C(F_1, F_2) & = (\reach(I_2; F_2) \setminus \reach(I_1; F_1)) \cap \reach(S_1; F_1). 
	\end{split}
	\]
	The score $\diffcover(S_1, S_2)$ can be rewritten as
	\[
	\sum_{F_1, F_2} p(F_1, F_2) (\abs{A(F_1, F_2)} + \abs{B(F_1, F_2)} + \abs{C(F_1, F_2)}),
	\]
	where $p(F_1, F_2)$ is the probability of $F_1$ being the realization
	of the edges for the first campaign
	and $F_2$ being the realization
	of the edges for the second campaign.
	
	The first term $A(F_1, F_2)$ does not depend on $S_1$ or $S_2$.
	The second term is submodular as a function of $S_2$ and does not depend of $S_1$.
	The third term is submodular as a function of $S_1$ and does not depend of $S_2$.
	Since any linear combination of submodular function weighted by positive coefficients
	is also submodular, this completes the proof.
\end{proof}

We are ready to discuss our algorithms.

\para{Algorithm 0: ignore $\diffnew$.}
Our first algorithm is very simple: instead of maximizing $\diff$, we maximize
$\diffcover$, i.e., we ignore any vertices that are made imbalanced
during the process.  
Since $\diffcover$ is submodular and monotone
we can use the greedy algorithm.  
If we then compare the obtained result
with the empty solution, we get the promised approximation guarantee.
We refer to this algorithm as \coveralg.

\begin{proposition}
\label{lemma:heterogeneous-expected-approximation}
Let $\brak{S_1^*, S_2^*}$ be the optimal solution maximizing $\diff$.
Let $\brak{S_1, S_2}$ be the solution obtained via greedy algorithm maximizing $\diffcover$.
Then
\[
	\max\{\diff(S_1, S_2), \diff(\emptyset, \emptyset)\} \geq 
	\frac{1 - 1/e}{2}\diff(S_1^*, S_2^*).
\]
\end{proposition}
\begin{proof}
	Write $c = 1 - 1/e$.
	Let $\brak{S_1', S_2'}$ be the optimal solution maximizing $\diffcover$.
	Lemma~\ref{lemma:static-joint-common-budget} shows that $\diffcover(S_1, S_2) \geq c\diffcover(S_1', S_2')$.
	
	Note that $\diffnew(\emptyset, \emptyset) \geq \diffnew(S_1^*, S_2^*)$
	as the first term is the average of vertices not affected by the initial seeds.
	Thus,
	\[
	\begin{split}
	\diff(S_1^*, S_2^*) & = \diffcover(S_1^*, S_2^*) + \diffnew(S_1^*, S_2^*) 
	\leq \diffcover(S_1', S_2') + \diffnew(S_1^*, S_2^*) \\
	& \leq \diffcover(S_1', S_2') + \diffnew(\emptyset, \emptyset) 
	\leq \diffcover(S_1, S_2) / c + \diffnew(\emptyset, \emptyset) \\
	& \leq \diffcover(S_1, S_2) / c + \diffnew(\emptyset, \emptyset) / c \\
	& \leq (2/c)\, \max\{\diffcover(S_1, S_2),  \diffnew(\emptyset, \emptyset)\} \\
	& \leq (2/c)\, \max\{\diff(S_1, S_2),  \diff(\emptyset, \emptyset)\}, \\
	\end{split}
	\]
	which completes the proof.
\end{proof}

\para{Algorithm 1: force common seeds.}
Ignoring the $\diffnew$ term may prove costly as it is possible to introduce a lot
of new imbalanced vertices. The idea behind the second algorithm is to
force $\diffnew = 0$. We do this by either adding the same seeds to both campaigns,
or adding a seed that is covered by an opposing campaign.
This algorithm has guarantees only in the correlated setting with even budget $k$
but in practice we can use the algorithm also for the heterogeneous setting. 
We refer to this algorithm as \commonalg\ and the pseudo-code is given in Algorithm~\ref{alg:greedy2}.

\begin{algorithm}
\caption{\commonalg, greedy algorithm that only adds common seeds}
\label{alg:greedy2}
	$S_1 \define S_2 \define \emptyset$\;
	\While {$\abs{S_1} + \abs{S_2} \leq k$} {
		$c \define \arg \max_c \diff(S_1 \cup \set{c}, S_2 \cup \set{c})$\;
		$s_1 \define \arg \max_{s \in I_1} \diff(S_1, S_2 \cup \set{s})$\;
		$s_2 \define \arg \max_{s \in I_2} \diff(S_1 \cup \set{s}, S_2)$\;
		add the best option among $\brak{c, c}$, $\brak{\emptyset, s_1}$, $\brak{s_2, \emptyset}$
		to $\brak{S_1, S_2}$ while respecting the budget.
	}
\end{algorithm}

We first show in the following lemma that adding common seeds may halve the score, in the worst case. Then, we use this lemma to prove the approximation guarantee

\begin{lemma}
\label{lemma:bounded-solution-homogeneous-common-seeds}
Let $\brak{S_1, S_2}$ be a solution to
\balanceprobc, with an even budget $k$.
There exists a solution $\brak{S_1', S_2'}$ with $S_1' = S_2'$ such
that $ \difftc(S_1', S_2') \geq \difftc(S_1, S_2) /2$.

\end{lemma}
\begin{proof}
	As we are dealing with the correlated setting, we can
	write $\reach(S) = \reachone(S) = \reachtwo(S)$.
	Our first step is to decompose $\omega = \difftc(S_1, S_2)$ into several components.
	To do so,
	we partition the vertices based on their reachability from the initial seeds,
	\begin{align*}
	A & = \reach(I_1) \cap \reach(I_2), &
	B & = \reach(I_1) \setminus \reach(I_2), \\
	C & = \reach(I_2) \setminus \reach(I_1), &
	D & = V \setminus (\reach(I_1) \cup \reach(I_2)) .
	\end{align*}
	Note that these are all random variables.
	If $S_1 = S_2 = \emptyset$, then $\difftc(S_1, S_2) = \meantc{\abs{A} + \abs{D}}$.
	More generally, $S_1$ may balance some vertices in $C$, and $S_2$ may balance some vertices in $B$.
	We may also introduce new imbalanced vertices in $D$. To take this into account we define
	\begin{align*}
	B' & = B \cap \reach(S_2), &
	C' & = C \cap \reach(S_1), \\
	D' & = D \setminus (\reach(S_1) \xor \reach(S_2)).
	\end{align*}
	We can express the cost of $\difftc(S_1, S_2)$ as
	\[
	\omega = \difftc(S_1, S_2) = \meantc{\abs{A} + \abs{B'} + \abs{C'} + \abs{D'}}.
	\]
	Split $S_1 \cup S_2$ in two equal-size sets, $T$ and $Q$,
	and define
	\[
	\omega_1 = \difftc(T, T), \quad \omega_2 = \difftc(Q, Q).
	\]
	We claim that $\omega \leq \omega_1 + \omega_2$. This proves the proposition, since
	$\omega_1 + \omega_2 \leq 2\max \{\omega_1, \omega_2\}$.
	
	To prove the claim let us first split $T$ and $Q$,
	\[
	T_1 = T \cap S_1,\ 
	T_2 = T \cap S_2,\ 
	Q_1 = Q \cap S_1,\ 
	Q_2 = Q \cap S_2.
	\]
	Our next step is to decompose $\omega_1$ and $\omega_2$,
	similar to $\omega$.  To do that, we define
	\begin{align*}
	B_1 &= B \cap \reach(T_2), &
	B_2 &= B \cap \reach(Q_2), \\
	C_1 &= C \cap \reach(T_1), &
	C_2 &= C \cap \reach(Q_1).
	\end{align*}
	Note that, the pair $\brak{T, T}$ does not introduce new imbalanced nodes. 
	This leads to
	\[
	\omega_1 = \difftc(T, T) = \meantc{\abs{A} + \abs{B_1} + \abs{C_1} + \abs{D}},
	\]
	and similarly,
	\[
	\omega_2 = \difftc(Q, Q) = \meantc{\abs{A} + \abs{B_2} + \abs{C_2} + \abs{D}}.
	\]
	To prove $\omega \leq \omega_1 + \omega_2$, note that $\abs{D'} \leq \abs{D}$.
	In addition,
	\[
	\begin{split}
	\abs{B'} & = \abs{B \cap (\reach(T_2) \cup \reach(Q_2))} \\
	& \leq  \abs{B \cap \reach(T_2)} + \abs{B \cap \reach(Q_2)} = \abs{B_1} + \abs{B_2}
	\end{split}
	\]
	and
	\[
	\begin{split}
	\abs{C'} & = \abs{C \cap (\reach(T_1) \cup \reach(Q_1))} \\
	& \leq  \abs{C \cap \reach(T_1)} + \abs{C \cap \reach(Q_1)} = \abs{C_1} + \abs{C_2}.
	\end{split}
	\]
	Combining these inequalities proves the proposition.
\end{proof}

It is easy to
see that the greedy algorithm satisfies the conditions of the following
proposition.

\begin{proposition}
\label{prop:greedy}
Assume an iterative algorithm where at each iteration,
we add one or two vertices to our solution until our constraints are met.
Let $S_1^i$, $S_2^i$ be the sets after the $i$-th iteration, $S_1^0 = S_2^0 = \emptyset$.
Let $\eta_i = \difftc(S_1^i, S_2^i)$ be the cost after the $i$-th iteration.
Assume that $\eta_{i} \geq \eta_{i - 1}$. 
Assume further that for $i = 1, \ldots, k / 2$ it holds that
$
	\eta_i \geq \difftc(S_1^{i - 1} \cup \set{c}, S_2^{i - 1} \cup \set{c}).
$
Then the algorithm yields $(1 - 1/e)/2$ approximation.
\end{proposition}
To prove the proposition, we need the following technical lemma, which is a twist
of a standard technique for proving the approximation ratio of the greedy algorithm on submodular functions.

\begin{lemma}
	\label{lem:partialgreedy}
	Assume a universe $U$. Let $\funcdef{f}{2^U}{\real}$ be a positive function.
	Let $T \subseteq U$ be a set with $k$ elements.
	Let $C_0 \subseteq \cdots \subseteq C_k$ be a sequence of subsets of $U$.
	Assume that $f(C_i) \geq \max_{t \in T} f(C_{i - 1} \cup \set{t})$.
	
	Assume further that for each $i = 1, \ldots, k$, we can decompose $f$ as $f = g_i + h_i$ such
	that
	\begin{enumerate}
		\item $g_i$ is submodular and monotonically increasing function,
		\item $h_i(W) = h_i(C_{i - 1})$, for any $W \subseteq T \cup C_{i - 1}$.
	\end{enumerate}
	
	Then $f(C_k) \geq (1 - 1/e) f(T)$.
\end{lemma}

\begin{proof}
	The assumptions of the propositions imply
	\[
	\begin{split}
	f(T) & = g_i(T) + h_i(T) \\
	& = g_i(T) + h_i(C_{i - 1}) \\
	& \leq g_i(C_{i - 1}) + h_i(C_{i - 1}) + \sum_{t \in T} g_i(C_{i - 1} \cup \set{t}) - g_i(C_{i - 1}) \\
	& = f(C_{i - 1}) + \sum_{t \in T} h_i(C_{i - 1}) + g_i(C_{i - 1} \cup \set{t}) - g_i(C_{i - 1}) - h_i(C_{i - 1}) \\
	& = f(C_{i - 1}) + \sum_{t \in T} h_i(C_{i - 1} \cup \set{t}) + g_i(C_{i - 1} \cup \set{t}) - g_i(C_{i - 1}) - h_i(C_{i - 1}) \\
	& = f(C_{i - 1}) + \sum_{t \in T} f(C_{i - 1} \cup \set{t}) - f(C_{i - 1}) \\
	& \leq f(C_{i - 1}) + k (f(C_{i}) - f(C_{i - 1})),
	\end{split}
	\]
	where the first inequality is due to the submodularity of $g_i$, and is a standard trick to prove the approximation
	ratio for the greedy algorithm.
	
	We can rewrite the above inequality as
	\[
	kf(T) + (1 - k)f(T) = f(T) \leq f(C_{i - 1}) + k (f(C_{i}) - f(C_{i - 1})).
	\]
	Rearranging the terms leads to
	\[
	\frac{k - 1}{k}(f(C_{i - 1}) - f(T)) \leq f(C_{i}) - f(T) \quad.
	\]
	Applying induction over $i$, yields
	\[
	f(C_{k}) - f(T) \geq \pr{\frac{k - 1}{k}}^k(f(C_{0}) - f(T)) \geq \frac{1}{e} (f(C_{0}) - f(T)) \geq - f(T) / e,
	\]
	leading to $f(C_k) \geq (1 - 1/e) f(T)$.
\end{proof}

We can now prove the main claim.
Note that since we are using the correlated model, we have $\reachone = \reachtwo$. For notational simplicity, we will write 
$\reach = \reachone = \reachtwo$.

\begin{proof}[Proof of Proposition~\ref{prop:greedy}]
	Let $\mathit{OPT}$ be the cost of the optimal solution.
	Let $D$ be the solution maximizing $\difftc(D, D)$ with $\abs{D} \leq k / 2$.
	Lemma~\ref{lemma:bounded-solution-homogeneous-common-seeds} guarantees that $\mathit{OPT} /2 \leq \difftc(D, D)$.
	
	In order to apply Lemma~\ref{lem:partialgreedy}, we first define
	the universe $U$ as
	\[
	U = \set{\brak{u, v} \mid u, v \in V} \cup \set{\brak{v, \emptyset} \mid v \in V} \cup \set{\brak{\emptyset, v} \mid v \in V}.
	\]
	
	The sets are defined as 
	\[
	C_i = \set{\brak{v, \emptyset} \mid v \in S^i_1} \cup \set{\brak{\emptyset, v} \mid v \in S^i_2}.
	\]
	
	Given a set $C \subseteq U$, let us define $\pi_1(C) = \set{v \mid \brak{v, u} \in C, v \neq \emptyset}$ to be the union of 
	the first entries in $C$. Similarly, define $\pi_2(C) = \set{v \mid \brak{u, v} \in C, v \neq \emptyset}$.
	
	We can now define $f$ as $f(C) = \difftc(\pi_1(C), \pi_2(C))$.
	To decompose $f$,
	let us first write
	\[
	X_i = \reach(I_1 \cup \pi_1(C_{i - 1})) \cup \reach(I_2 \cup \pi_2(C_{i - 1})) = \reach(I_1 \cup S^{i - 1}_1) \cup \reach(I_2 \cup S^{i - 1}_2), \quad Y_i = V \setminus X_i.
	\]
	and set
	\[
	\begin{split}
	g_i(C) & = \mean{}{\abs{X_i \setminus (\reach(I_1 \cup \pi_1(C)) \xor \reach(I_2 \cup \pi_2(C)))}}, \\
	h_i(C) & = \mean{}{\abs{Y_i \setminus (\reach(I_1 \cup \pi_1(C)) \xor \reach(I_2 \cup \pi_2(C)))}}.
	\end{split}
	\]
	Finally, we set $T = \set{\brak{d, d} \mid d \in D}$.
	
	First note that $f = g_i + h_i$ since $X_i \cap Y_i = \emptyset$.
	The proof of Lemma~\ref{lemma:static-joint-common-budget}
	shows that $g_i$ is monotonically increasing and submodular. 
	
	Let $C \subseteq C_{i - 1} \cup T$. If there is a vertex $v$ in
	$\reach(I_1 \cup \pi_1(C))$ but not in $X_i$, then this means $v$ was influenced
	by $d \in D$. Since $d \in \pi_2(C)$, we have $v \in \reach(I_2 \cup \pi_2(C))$.
	That is,
	\[
	\reach(I_1 \cup \pi_1(C)) \setminus X_i = \reach(I_2 \cup \pi_2(C)) \setminus X_i.
	\]
	Since $Y_i$ and $X_i$ are disjoint, this gives us
	\[
	\begin{split}
	h_i(C)
	& = \mean{}{\abs{Y_i \setminus (\reach(I_1 \cup \pi_1(C)) \xor \reach(I_2 \cup \pi_2(C)))}}  \\
	& = \mean{}{\abs{Y_i \setminus ((\reach(I_1 \cup \pi_1(C)) \setminus X_i) \xor (\reach(I_2 \cup \pi_2(C)) \setminus X_i))}}  \\
	& = \mean{}{\abs{Y_i}} .
	\end{split}
	\]
	
	That is, $h_i(C)$ is constant for any $C \subseteq C_{i - 1} \cup T$.
	Thus, $h_i(C) = h_i(C_{i - 1})$.
	
	Finally, the assumpion of the proposition guarantees that $f(C_i) \geq f(C_{i - 1} \cup \set{t})$, for $t \in T$.
	
	Thus, these definitions meet all the prerequisites of Lemma~\ref{lem:partialgreedy}, guaranteeing
	that
	\[
	(1 - 1/e)\difftc(D, D) \leq \difftc(S_1^{k/2}, S_2^{k/2}) \leq \difftc(S_1^{k}, S_2^{k}).
	\]
	Since $\mathit{OPT} /2 \leq \difftc(D, D)$, the result follows.
\end{proof}

\para{Algorithm 2: common seeds as baseline.}
Not allowing new imbalanced vertices may prove to be too restrictive.
We can relax this condition by allowing new imbalanced vertices as long as 
the gain is at least as good as adding a common seed. 
We refer to this algorithm as Hedge and the pseudo-code is given in Algorithm~\ref{alg:greedy3}.
The approximation guarantee for this algorithm---in the correlated setting 
and with even budget---follows immediately from
Proposition~\ref{prop:greedy} as it also satisfies the conditions.

\begin{algorithm}
\caption{\hedgealg, greedy algorithm, where each step is as good as adding the best common seed}
\label{alg:greedy3}
	$S_1 \define S_2 \define \emptyset$\;
	\While {$\abs{S_1} + \abs{S_2} \leq k$} {
		$c \define \arg \max_c \diff(S_1 \cup \set{c}, S_2 \cup \set{c})$\;
		$s_1 \define \arg \max_{s} \diff(S_1, S_2 \cup \set{s})$\;
		$s_2 \define \arg \max_{s} \diff(S_1 \cup \set{s}, S_2)$\;
		add the best option among $\brak{c, c}$, $\brak{\emptyset, s_1}$, $\brak{s_2, \emptyset}$, $\brak{s_2, s_1}$,
		to $\brak{S_1, S_2}$ while respecting the budget.
	}
\end{algorithm}


\section{Experimental evaluation}
\label{sec:experiments}

In this section, 
we evaluate the effectiveness of our algorithms on real-world datasets. 
We focus on 
($i$) analyzing the quality of the seeds picked by our algorithms in comparison 
to other heuristic approaches and baselines;
($ii$) analyzing the efficiency and the scalability of our algorithms; 
and 
($iii$) providing anecdotal examples of the obtained results. 

For all experiments
we report averages over $1\,000$ 
random simulations of the cascade process.
As argued by~\citet{kempe2003maximizing}, 
a random cascade can be generated in advance by sampling each edge $e$
with probability $\pone(e)$ or $\ptwo(e)$, 
depending on the model and the campaign.
In all experiments we set $k$ to range between $5$ and $50$ with a step of $5$.

\spara{Datasets:}
To evaluate the effectiveness of our algorithms, 
we run experiments on real-world data collected from twitter.
Let $G=(V,E)$ be the twitter follower graph.
A directed edge $(u,v)\in E$ indicates that user $v$ follows $u$; 
note that the edge direction indicates the ``information flow''
from a user to their followers.
We define a cascade $G_X=(X,E_X)$ as a graph over the set of users $X\subseteq V$ 
who have retweeted at least one hashtag related to a topic (e.g., US elections).
An edge $(u,v)\in E_X\subseteq E$ indicates that $v$ retweeted $u$.

We use datasets from six topics with opposing viewpoints, covering 
politics ({\uselections}, {\brexit}, {\obamacare}), 
policy ({\abortion, \fracking}), and 
lifestyle ({\iphone}, focusing on iPhone vs.\ Samsung).
All datasets are collected by filtering the twitter streaming API 
(1\% random sample of all tweets) 
for a set of keywords used in previous work~\cite{lu2015biaswatch}. 
For each dataset, we identify two sides 
(indicating the two view-points) on the retweet graph, 
which has been shown to capture best the two opposing sides 
of a controversy~\cite{garimella2016quantifying}.
Details on the statistics of the dataset are shown in Table~\ref{tab:datasets}.

\begin{table}
	\setlength{\tabcolsep}{1.3mm}
	\centering
	\caption{\label{tab:datasets}Dataset statistics. The column $|C|$ refers to the average number of edges in a randomly generated cascade in the correlated case, while $|C_1|$ and $|C_2|$ refer to average number of edges generated in a cascade of the campaigns $1$ and $2$, respectively, in the heterogeneous case.}
	\begin{tabular}{@{}lrrrrr}
		\toprule
		Dataset & \# Nodes & \# Edges & $|C|$ & $|C_1|$ &  $|C_2|$ \\
		\midrule
		\abortion & 279\,505 & 671\,144  & 2\,105 & 326 & 1\,801  \\
		\brexit & 22\,745 & 48\,830  & 476 & 113  & 390 \\
		\fracking & 374\,403 & 1\,377\,085  & 4\,156 & 1\,595 & 3\,103  \\
		\iphone & 36\,742 & 49\,248  & 4\,776 & 339 & 4\,478  \\
		\obamacare & 334\,617 & 1\,511\,670 & 6\,614 & 2\,404  & 4\,527  \\
		\uselections & 80\,544 & 921\,368 & 4\,697 & 3\,097 & 12\,044 \\
		\bottomrule
	\end{tabular}
\end{table}

After building the graphs, we need to estimate the diffusion probabilities 
for the heterogeneous and correlated models.
Note that the estimation of the diffusion probabilities 
is orthogonal to our contribution in this paper. 
For the sake of concreteness we have used the approach described below. 
One could use a different, more advanced, method; 
our methods are still applicable.

Let $\sideone(v)$ and $\sidetwo(v)$ 
be an {\em a priori} probability of a user $v$ retweeting sides 1 and 2, respectively.
These are measured from the data by looking at how often 
a user retweets content from users and keywords 
that are discriminative of each side.
For example, for \uselections, the discriminative users and keywords for side Hillary would be 
@hillaryclinton and \#imwither, and for 
Trump, @realdonaldtrump and \#makeamericagreatagain.
Additional details on data collection, and user/keyword sets for 
side identification are given in Table~\ref{tab:description} in Appendix~\ref{sec:expextra}.

The probability that user $v$ retweets user $u$ (cascade probability) is
then defined~as
\begin{equation*}
\pii(u,v) = 
\alpha\,  \sideii(v) + (1-\alpha) \left(\frac{\retweets(u,v)+1}{\retweets(v)+2}\right), 
\quad i=1,2,
\end{equation*}
where $\retweets(u,v)$ is the number of times $v$ has retweeted $u$, and 
$\retweets(v)$ is the total number of retweets of user $v$. 
The cascade probabilities $\pii$ capture the fact that 
users retweet content 
if they see it from their friends 
(term $\frac{\retweets(u,v)+1}{\retweets(v)+2}$)
or based on their own biases (term~$\sideii(v)$).
The additive terms in the numerator and denominator 
provide an additive smoothing by Laplace's rule of succession.

We set the value of $\alpha$ to 0.8 for the heterogeneous setting. For the correlated setting, $\alpha$ is set to zero.










\begin{figure}[t!]
\newlength{\imgwidth}
\setlength{\imgwidth}{3cm}
\newlength{\imgwidthhet}
\setlength{\imgwidthhet}{3cm}
\setlength{\tabcolsep}{1pt}
\begin{tabular}{rrr}
\begin{tikzpicture}
\begin{axis}[xlabel={budget $k$},ylabel= {symm. diff.},
    width = \imgwidthhet,
    height = 2cm,
	title = {\iphone},
    cycle list name=yaf,
    scale only axis,
    x tick label style = {/pgf/number format/set thousands separator = {\,}},
    y tick label style = {/pgf/number format/set thousands separator = {\,}},
    scaled ticks = false,
    xtick = {10, 20, 30, 40, 50},
	every axis plot post/.append style= {line width=1pt},
    ]
\addplot table[x index = 0, y index = 1, header = true] {heterogeneous_experimental_data_iphone_samsung.txt};
\addplot table[x index = 0, y index = 2, header = true] {heterogeneous_experimental_data_iphone_samsung.txt};
\addplot table[x index = 0, y index = 4, header = true] {heterogeneous_experimental_data_iphone_samsung.txt};
\addplot table[x index = 0, y index = 3, header = true] {heterogeneous_experimental_data_iphone_samsung.txt};
\pgfplotsextra{\yafdrawaxis{5}{50}{295}{682}}
\end{axis}
\end{tikzpicture} &
\begin{tikzpicture}
\begin{axis}[xlabel={budget $k$},ylabel= {symm. diff.},
    width = \imgwidthhet,
    height = 2cm,
	title = {\obamacare},
    cycle list name=yaf,
    scale only axis,
    x tick label style = {/pgf/number format/set thousands separator = {\,}},
    y tick label style = {/pgf/number format/set thousands separator = {\,}},
    scaled ticks = false,
    xtick = {10, 20, 30, 40, 50},
	every axis plot post/.append style= {line width=1pt},
    ]
\addplot table[x index = 0, y index = 1, header = true] {heterogeneous_experimental_data_obamacare.txt};
\addplot table[x index = 0, y index = 2, header = true] {heterogeneous_experimental_data_obamacare.txt};
\addplot table[x index = 0, y index = 4, header = true] {heterogeneous_experimental_data_obamacare.txt};
\addplot table[x index = 0, y index = 3, header = true] {heterogeneous_experimental_data_obamacare.txt};
\pgfplotsextra{\yafdrawaxis{5}{50}{1235}{2648}}
\end{axis}
\end{tikzpicture} &
\begin{tikzpicture}
\begin{axis}[xlabel={budget $k$},ylabel= {symm. diff.},
    width = \imgwidthhet,
    height = 2cm,
	every axis plot post/.append style= {line width=1pt},
    cycle list name=yaf,
    scale only axis,
	title = {\uselections},
    x tick label style = {/pgf/number format/set thousands separator = {\,}},
    y tick label style = {/pgf/number format/set thousands separator = {\,}},
    scaled ticks = false,
    xtick = {10, 20, 30, 40, 50},
    legend style={nodes={scale=0.75, transform shape}},
	legend entries = {\coveralg, \hedgealg, \commonalg, \greedyalg}
    ]
\addplot table[x index = 0, y index = 1, header = true] {heterogeneous_experimental_data_uselections.txt};
\addplot table[x index = 0, y index = 2, header = true] {heterogeneous_experimental_data_uselections.txt};
\addplot table[x index = 0, y index = 4, header = true] {heterogeneous_experimental_data_uselections.txt};
\addplot table[x index = 0, y index = 3, header = true] {heterogeneous_experimental_data_uselections.txt};
\pgfplotsextra{\yafdrawaxis{5}{50}{231}{1934}}
\end{axis}
\end{tikzpicture} \\[-3mm]

\begin{tikzpicture}
\begin{axis}[xlabel={budget $k$},ylabel= {symm. diff.},
    width = \imgwidth,
    height = 2cm,
	title = {\iphone},
    cycle list name=yaf,
    scale only axis,
    x tick label style = {/pgf/number format/set thousands separator = {\,}},
    y tick label style = {/pgf/number format/set thousands separator = {\,}},
    scaled ticks = false,
    xtick = {10, 20, 30, 40, 50},
	every axis plot post/.append style= {line width=1pt},
    ]
\addplot table[x index = 0, y index = 1, header = true] {homogeneous_experimental_data_iphone_samsung.txt};
\addplot table[x index = 0, y index = 2, header = true] {homogeneous_experimental_data_iphone_samsung.txt};
\addplot table[x index = 0, y index = 4, header = true] {homogeneous_experimental_data_iphone_samsung.txt};
\addplot table[x index = 0, y index = 3, header = true] {homogeneous_experimental_data_iphone_samsung.txt};
\pgfplotsextra{\yafdrawaxis{5}{50}{0}{91}}
\end{axis}
\end{tikzpicture} &
\begin{tikzpicture}
\begin{axis}[xlabel={budget $k$},ylabel= {symm. diff.},
    width = \imgwidth,
    height = 2cm,
	title = {\obamacare},
    cycle list name=yaf,
    scale only axis,
    x tick label style = {/pgf/number format/set thousands separator = {\,}},
    y tick label style = {/pgf/number format/set thousands separator = {\,}},
    scaled ticks = false,
    xtick = {10, 20, 30, 40, 50},
	every axis plot post/.append style= {line width=1pt},
    ]
\addplot table[x index = 0, y index = 1, header = true] {homogeneous_experimental_data_obamacare.txt};
\addplot table[x index = 0, y index = 2, header = true] {homogeneous_experimental_data_obamacare.txt};
\addplot table[x index = 0, y index = 4, header = true] {homogeneous_experimental_data_obamacare.txt};
\addplot table[x index = 0, y index = 3, header = true] {homogeneous_experimental_data_obamacare.txt};
\pgfplotsextra{\yafdrawaxis{5}{50}{2}{490}}
\end{axis}
\end{tikzpicture} &
\begin{tikzpicture}
\begin{axis}[xlabel={budget $k$},ylabel= {symm. diff.},
    width = \imgwidth,
    height = 2cm,
	every axis plot post/.append style= {line width=1pt},
    cycle list name=yaf,
    scale only axis,
	title = {\uselections},
    x tick label style = {/pgf/number format/set thousands separator = {\,}},
    y tick label style = {/pgf/number format/set thousands separator = {\,}},
    scaled ticks = false,
    xtick = {10, 20, 30, 40, 50},
    ]
\addplot table[x index = 0, y index = 1, header = true] {homogeneous_experimental_data_uselections.txt};
\addplot table[x index = 0, y index = 2, header = true] {homogeneous_experimental_data_uselections.txt};
\addplot table[x index = 0, y index = 4, header = true] {homogeneous_experimental_data_uselections.txt};
\addplot table[x index = 0, y index = 3, header = true] {homogeneous_experimental_data_uselections.txt};
\pgfplotsextra{\yafdrawaxis{5}{50}{0}{2208}}
\end{axis}
\end{tikzpicture} \\[-1mm]
\end{tabular}
\caption{Expected symmetric difference $n - \difftc$ as a function of the budget $k$. Top row, heterogeneous model, bottom row: Correlated model. Low values are better.}
\label{fig:score}

\vspace{-.15cm}
\end{figure}
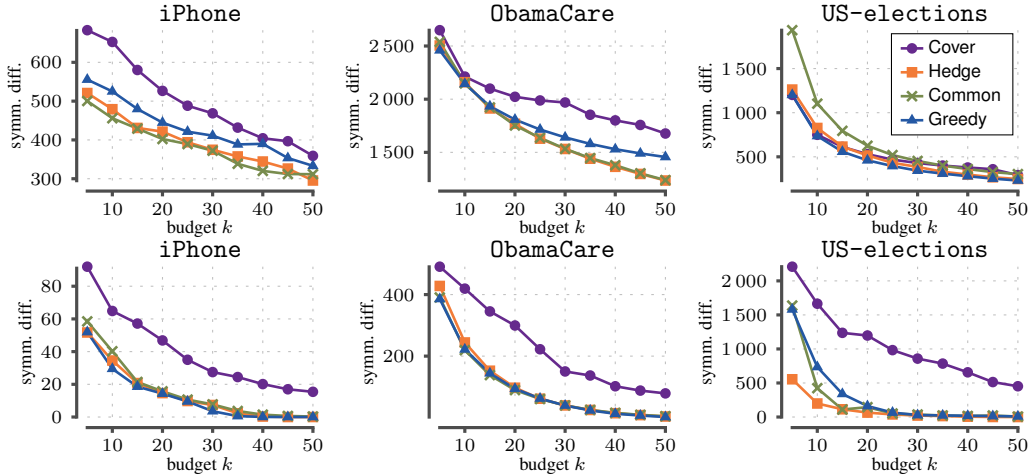

\spara{Baselines.}
We use 5 different baselines.
The first baseline, \borodinalg, is an adaptation of the framework by~\citet{borodin2017strategyproof}. 
This framework requires an objective function as an input, and here we use our objective function $\diff$.
The framework works as follows:
The two campaigns are given a budget $k / 2$ on the number of seeds that they can select.
At each round,
we select a vertex $v$ optimizing $\diff(S_1 \cup \set{v}, S_2)$, the vertex $v$
is added to $S_1$,
and then
we select a vertex $w$ optimizing $\diff(S_1, S_2 \cup \set{w})$, which is added to $S_2$.
We should stress that
the theoretical guarantees by~\citep{borodin2017strategyproof} do not apply because
our objective is not submodular. Thus, \borodinalg is a heuristic.



The next two heuristics add a set of common seeds to both campaigns.
We run a greedy algorithm for campaign $i = 1, 2$ to select the set $S'_i$ with the $\ell \gg k$ vertices $P_i$ that optimizes the function $\reach_i(S_i' \cup I_i)$.
We consider two heuristics:
\unionalg selects $S_1$ and $S_2$ to be equal to the $k/2$ first distinct vertices in $S_1' \cup S_2'$ while 
\interalg selects $S_1$ and $S_2$ to be equal to $k/2$ first vertices in $S_1' \cap S_2'$. 
Here the vertices are ordered based on their discovery time.
In both cases we (arbitrarily) break ties in favor of campaign~1.

Finally, {\highdegreealg} selects the vertices with the largest number of followers and assigns them alternately to the two cascades; and 
{\randomalg} assigns $k/2$ random seeds to each campaign.

In addition to the baselines, 
we also consider a simple greedy algorithm {\greedyalg}. 
The difference between {\coveralg} and {\greedyalg}
is that, in each iteration, 
{\coveralg} adds the seed that maximizes~$\diffcover$, 
while 
{\greedyalg} adds the seed that maximizes~$\diff$.
We can only show an approximation guarantee for {\coveralg}
but {\greedyalg} is a more intuitive approach
as it considers directly the number of balanced vertices
and we use it as a heuristic.

\spara{Comparison of the algorithms.}
We start by evaluating the quality of the sets of seeds computed by our algorithms, 
i.e., the number of equally-informed vertices.

\smallskip
\noindent
\emph{Heterogeneous setting.}
We consider first the case of heterogeneous networks. 
The results for the selected datasets are shown in Figure~\ref{fig:score}. Full results are shown in Appendix~\ref{sec:expextra}.
Instead of plotting $\diff$,
we plot the number of the remaining unbalanced vertices, $n - \diff$,
as it makes the results easier to distinguish; 
i.e., an optimal solution achieves the value 0.

The first observation is that the approximation algorithm \coveralg performs, 
in general, worse than the other two heuristics.
This is due to the fact that 
{\coveralg} does not optimize directly the objective function.
\hedgealg performs better than {\greedyalg}, in general, 
since it examines additional choices to select.
The only deviation from this picture is for the {\uselections} dataset, 
where the {\greedyalg} outperforms {\hedgealg} by a small factor.
This may due to the fact that 
while {\hedgealg} has more options, 
it allocates seeds in batches of two. 

The algorithms seem to follow a diminishing-returns behavior,
on most cases.
This behavior appears despite the fact that the optimization 
function is not submodular.

\smallskip
\noindent
\emph{Correlated setting.}
Next we consider correlated networks. 
We experiment with the three 
approximation algorithms \coveralg, \commonalg, \hedgealg, 
and the heuristic {\greedyalg}. 
The results are shown in Figure~\ref{fig:score}.
\coveralg performs again the worst 
since it is the only method that introduces new unbalanced vertices 
without caring about their cardinality.
Its variant, {\greedyalg},  performs much better in practice
even though it does not provide an approximation guarantee.
The algorithms \commonalg, {\greedyalg}, and \hedgealg perform very similar 
to each other without a clear winner.


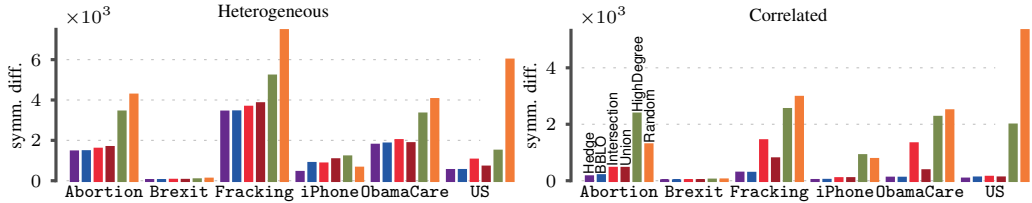
\begin{figure}
\begin{tikzpicture}
\begin{axis}[ylabel= {symm. diff.},
    width = 5.5cm,
    height = 2cm,
	title = {\scriptsize Heterogeneous},
    cycle list name=yaf,
    scale only axis,
    tick scale binop=\times,
    x tick label style = {/pgf/number format/set thousands separator = {\,}},
    y tick label style = {/pgf/number format/set thousands separator = {\,}},
    scaled y ticks = base 10:-3,
	xtick = \empty,
	every axis plot post/.append style= {line width=1pt},
	xmin = -0.5,
	ymin = 0,
	ybar,
	bar width = 2.5pt,
    no markers
    ]
\addplot[yafcolor1, fill] table[x expr = \coordindex, y index = 2, header = true] {heterogeneous_experimental_data_baseline.txt};
\addplot[yafcolor5, fill] table[x expr = \coordindex, y index = 8, header = true] {heterogeneous_experimental_data_baseline.txt};
\addplot[yafcolor4, fill] table[x expr = \coordindex, y index = 7, header = true] {heterogeneous_experimental_data_baseline.txt};
\addplot[yafcolor7, fill] table[x expr = \coordindex, y index = 9, header = true] {heterogeneous_experimental_data_baseline.txt};
\addplot[yafcolor3, fill] table[x expr = \coordindex, y index = 6, header = true] {heterogeneous_experimental_data_baseline.txt};
\addplot[yafcolor2, fill] table[x expr = \coordindex, y index = 5, header = true] {heterogeneous_experimental_data_baseline.txt};

\tikzstyle{labnode} = [inner sep = 1pt, font = \scriptsize, anchor = mid, yshift = -4pt] 

\node[labnode] at (axis cs: 0, 0) {\abortion};
\node[labnode] at (axis cs: 1, 0) {\brexit};
\node[labnode] at (axis cs: 2, 0) {\fracking};
\node[labnode] at (axis cs: 3, 0) {\iphone};
\node[labnode] at (axis cs: 4, 0) {\obamacare};
\node[labnode] at (axis cs: 5, 0) {\uselectionsalt};
\pgfplotsextra{\yafdrawyaxis{0}{7499}}
\end{axis}
\end{tikzpicture}%
\begin{tikzpicture}
\begin{axis}[ylabel= {symm. diff.},
    width = 5.5cm,
    height = 2cm,
	title = {\scriptsize Correlated},
    cycle list name=yaf,
    scale only axis,
    tick scale binop=\times,
    x tick label style = {/pgf/number format/set thousands separator = {\,}},
    y tick label style = {/pgf/number format/set thousands separator = {\,}},
    scaled y ticks = base 10:-3,
	xtick = \empty,
	every axis plot post/.append style= {line width=1pt},
	xmin = -0.5,
	ymin = 0,
	ybar,
	bar width = 2.5pt,
    no markers
    ]
\addplot[yafcolor1, fill] table[x expr = \coordindex, y index = 2, header = true] {homogeneous_experimental_data_baseline.txt};
\addplot[yafcolor5, fill] table[x expr = \coordindex, y index = 8, header = true] {homogeneous_experimental_data_baseline.txt};
\addplot[yafcolor4, fill] table[x expr = \coordindex, y index = 7, header = true] {homogeneous_experimental_data_baseline.txt};
\addplot[yafcolor7, fill] table[x expr = \coordindex, y index = 9, header = true] {homogeneous_experimental_data_baseline.txt};
\addplot[yafcolor3, fill] table[x expr = \coordindex, y index = 6, header = true] {homogeneous_experimental_data_baseline.txt};
\addplot[yafcolor2, fill] table[x expr = \coordindex, y index = 5, header = true] {homogeneous_experimental_data_baseline.txt};
\node[inner sep = 1pt, font = \tiny, anchor = west, rotate = 90, yshift = 11pt, scale = 0.6] at (axis cs: 0, 179) {\hedgealg};
\node[inner sep = 1pt, font = \tiny, anchor = west, rotate = 90, yshift = 7pt, scale = 0.6] at (axis cs: 0, 185) {\borodinalg};
\node[inner sep = 1pt, font = \tiny, anchor = west, rotate = 90, yshift = 2.5pt, scale = 0.6] at (axis cs: 0, 446) {\interalg};
\node[inner sep = 1pt, font = \tiny, anchor = west, rotate = 90, yshift = -2.5pt, scale = 0.6] at (axis cs: 0, 444) {\unionalg};
\node[inner sep = 1pt, font = \tiny, anchor = west, rotate = 90, yshift = -7.5pt, scale = 0.6] at (axis cs: 0, 2368) {\highdegreealg};
\node[inner sep = 1pt, font = \tiny, anchor = west, rotate = 90, yshift = -11.5pt, scale = 0.6] at (axis cs: 0, 1279) {\randomalg};

\tikzstyle{labnode} = [inner sep = 1pt, font = \scriptsize, anchor = mid, yshift = -4pt] 

\node[labnode] at (axis cs: 0, 0) {\abortion};
\node[labnode] at (axis cs: 1, 0) {\brexit};
\node[labnode] at (axis cs: 2, 0) {\fracking};
\node[labnode] at (axis cs: 3, 0) {\iphone};
\node[labnode] at (axis cs: 4, 0) {\obamacare};
\node[labnode] at (axis cs: 5, 0) {\uselectionsalt};

\pgfplotsextra{\yafdrawyaxis{0}{5325}}
\end{axis}
\end{tikzpicture}
\caption{Expected symmetric difference $n - \diff$ of \hedgealg and the baselines. $k = 20$. Low values are better.}
\label{fig:baseline}
\vspace{-.15cm}
\end{figure}

\spara{Comparison with baselines.}
Our next step is to compare against the baselines.
For simplicity, we focus on $k = 20$; the overall conclucions hold for other budgets.
The results for \hedgealg versus the five baselines are shown in Figure~\ref{fig:baseline}.

From the results we see that \borodinalg is the best competitor: its scores
are the closest to \hedgealg, and it receives slightly better scores in 3 out
of 12 cases. The competitiveness is not surprising because we specifically set 
the objective function in \borodinalg to be $\diff(S_1, S_2)$.
The \interalg and \unionalg also perform well but are always worse than \hedgealg.
{\randomalg} is unpredictable but always worse than {\hedgealg}.
In the case of heterogeneous networks, {\hedgealg} selects seeds that 
leave less unbalanced vertices, 
by a factor of two on average, compared to the seeds selected by the 
{\highdegreealg} method.
For correlated networks, our method outperforms the two baselines by an order of magnitude.

\spara{Running time.}
We proceed to evaluate the efficiency and the scalability of our algorithms. 
The running times, in seconds, of our algorithms, for all datasets and for $k=20$, 
are shown in Figure~\ref{fig:running_time} in Appendix~\ref{sec:expextra}
as a function of network size. 
We observe that all algorithms have comparable running times
and good scalability.

\spara{Use case with \fracking.}
We present a qualitative case-study analysis for the seeds selected by our algorithm.
We highlight the {\fracking} dataset, even though we applied similar analysis
to the other datasets as well (the results are given in Figure~\ref{fig:clouds} in Appendix~\ref{sec:expextra}).
Recall that for each dataset we identify two sides with opposing views, 
and a set of initial seeds for each side ($I_1$ and $I_2$).
We consider the users in the initial seeds $I_1$ (side supporting fracking),
and summarize the text of all their Twitter profile descriptions in a word cloud. 
The result, as can be seen in Figure~\ref{fig:clouds} in Appendix~\ref{sec:expextra},
contains words that are used to emphasize 
the benefits of fracking (energy, oil, gas, etc.).
We then draw a similar word cloud
for the users identified by the {\hedgealg} algorithm
as seed nodes in the sets $S_1$ and $S_2$ ($k=50$).
The result, shown in Figure~\ref{fig:clouds} in Appendix~\ref{sec:expextra},
contains a more balanced set of words, 
which includes many words used to underline the environmental 
dangers of fracking. 


\section{Conclusion}
\label{sec:conclusion}
We presented the first study of the problem of 
balancing information exposure in social networks
using techniques 
from the area of information diffusion.
Our approach has several novel aspects.
In particular, we formulate our problem by seeking to optimize a
{\em symmetric difference} function, 
which is neither monotone nor submodular, 
and thus, not amenable to existing approaches.
Additionally, while previous studies consider 
a setting with selfish agents
and provide bounds on best-response strategies
(i.e., move of the last player), 
we consider a centralized setting and
provide bounds for a global objective function.

Our work provides several directions for future work. 
One interesting problem is to 
improve the approximation guarantee 
for the problem we define.
Second, we would like to extend the problem definition 
for more than two campaigns and
design approximation algorithms for that case.

\subsubsection*{Acknowledgments}
Work partially done while Nikos Parotsidis was visiting Aalto University.
This work has been supported by the Academy of Finland project ``Nestor'' (286211) and the EC H2020 RIA project ``SoBigData'' (654024).

\bibliographystyle{abbrvnat}
\bibliography{burst-brief}

\newpage

\appendix


\section{Additional tables and figures related to the experimental evaluation}

\label{sec:expextra}

\begin{table}[h]
\caption{Dataset descriptions, as well as tags and rewteets that were used to collect the data.}
\label{tab:description}
\begin{tabular}{p{6.5cm}p{6.5cm}}
\toprule
\multicolumn{2}{l}{
\begin{minipage}{13.4cm}
\textbf{USelections}: Tweets containing hashtags and keywords identifying the USElections, such as 
\#uselections, \#trump2016, \#hillary2016, etc. Collected using Twitter 1\% sample for 2 weeks in September 2016
\end{minipage}
}\\[5mm]
\emph{Pro-Hillary} & \emph{Pro-Trump} \\
RT @hillaryclinton, \#hillary2016, \#clintonkaine2016, \#imwithher &
RT @realdonaldtrump, \#makeamericagreatagain, \#trumppence16, \#trump2016 \\

\midrule

\multicolumn{2}{l}{
\begin{minipage}{13cm}
\textbf{Brexit}: Tweets containing hashtags \#brexit, \#voteremain, \#voteleave, \#eureferendum for all of June 2016, from the 1\% Twitter sample.
\end{minipage}
}\\[5mm]

\emph{Pro-Remain} & \emph{Pro-Leave} \\
\#voteremain, \#strongerin, \#remain, \#remaineu, \#votein &
\#voteleave, \#strongerout, \#leaveeu, \#takecontrol, \#leave, \#voteout \\

\midrule

\multicolumn{2}{l}{
\begin{minipage}{13cm}
\textbf{Abortion}: Tweets containing hashtags \#abortion, \#prolife, \#prochoice, \#anti-abortion, \#pro-abortion, \#plannedparenthood from Oct 2011 to Aug 2016.
\end{minipage}
}\\[5mm]
\emph{Pro-Choice} & \emph{Pro-Life}\\
RT @thinkprogress, RT @komenforthecure, RT @mentalabortions, \#waronwomen, \#nbprochoice, \#prochoice, \#standwithpp, \#reprorights &
RT @stevenertelt, RT @lifenewshq, \#praytoendabortion, \#prolifeyouth, \#prolife, \#defundplannedparenthood, \#defundpp, \#unbornlivesmatter \\

\midrule

\multicolumn{2}{l}{
\begin{minipage}{13cm}
\textbf{Obamacare}: Tweets containing hashtags \#obamacare, and \#aca from Oct 2011 to Aug 2016.
\end{minipage}
}\\[5mm]
\emph{Pro-Obamacare} & \emph{Anti-Obamacare} \\
RT @barackobama, RT @lolgop, RT @charlespgarcia, RT @defendobamacare, RT @thinkprogress, \#obamacares, \#enoughalready, \#uniteblue &
RT @sentedcruz, RT @realdonaldtrump, RT @mittromney, RT @breitbartnews, RT @tedcruz, \#defundobamacare, \#makedclisten, \#fullrepeal, \#dontfundit \\

\midrule

\multicolumn{2}{l}{
\begin{minipage}{13cm}
\textbf{Fracking}: Tweets containing hashtags and keywords \#fracking, 'hydraulic fracturing', 'shale', 'horizontal drilling', from Oct 2011 to Aug 2016.
\end{minipage}
}\\[5mm]
\emph{Pro-Fracking} & \emph{Anti-Fracking} \\
RT @shalemarkets, RT @energyindepth, RT @shalefacts, \#fracknation, \#frackingez, \#oilandgas, \#greatgasgala, \#shalegas &
RT @greenpeaceuk, RT @greenpeace, RT @ecowatch, \#environment, \#banfracking, \#keepitintheground, \#dontfrack, \#globalfrackdown, \#stopthefrackattack \\

\midrule

\multicolumn{2}{l}{
\begin{minipage}{13cm}
\textbf{iPhone vs. Samsung}: Tweets containing hashtags \#iphone, and \#samsung from April (release of Samsung Galaxy S7), and September 2015 (release of iPhone 7).
\end{minipage}
}\\[5mm]
\emph{Pro-iPhone} & \emph{Pro-Samsung} \\
\#iphone &
\#samsung \\
\bottomrule
\end{tabular}
\end{table}

\begin{figure}[h]
\setlength{\imgwidthhet}{2.9cm}
\begin{tikzpicture}
\begin{axis}[xlabel={budget $k$},ylabel= {symm. diff.},
    width = \imgwidthhet,
    height = 2cm,
	title = {\abortion},
    cycle list name=yaf,
    scale only axis,
    x tick label style = {/pgf/number format/set thousands separator = {\,}},
    y tick label style = {/pgf/number format/set thousands separator = {\,}},
    scaled ticks = false,
    xtick = {10, 20, 30, 40, 50},
	every axis plot post/.append style= {line width=1pt},
    ]
\addplot table[x index = 0, y index = 1, header = true] {heterogeneous_experimental_data_abortion.txt};
\addplot table[x index = 0, y index = 2, header = true] {heterogeneous_experimental_data_abortion.txt};
\addplot table[x index = 0, y index = 4, header = true] {heterogeneous_experimental_data_abortion.txt};
\addplot table[x index = 0, y index = 3, header = true] {heterogeneous_experimental_data_abortion.txt};
\pgfplotsextra{\yafdrawaxis{5}{50}{956}{2702}}
\end{axis}
\end{tikzpicture}%
\begin{tikzpicture}
\begin{axis}[xlabel={budget $k$},ylabel= {symm. diff.},
    width = \imgwidthhet,
    height = 2cm,
	title = {\brexit},
    cycle list name=yaf,
    scale only axis,
    x tick label style = {/pgf/number format/set thousands separator = {\,}},
    y tick label style = {/pgf/number format/set thousands separator = {\,}},
    scaled ticks = false,
    xtick = {10, 20, 30, 40, 50},
	every axis plot post/.append style= {line width=1pt},
    ]
\addplot table[x index = 0, y index = 1, header = true] {heterogeneous_experimental_data_brexit.txt};
\addplot table[x index = 0, y index = 2, header = true] {heterogeneous_experimental_data_brexit.txt};
\addplot table[x index = 0, y index = 4, header = true] {heterogeneous_experimental_data_brexit.txt};
\addplot table[x index = 0, y index = 3, header = true] {heterogeneous_experimental_data_brexit.txt};
\pgfplotsextra{\yafdrawaxis{5}{50}{7}{54}}
\end{axis}
\end{tikzpicture}%
\begin{tikzpicture}
\begin{axis}[xlabel={budget $k$},ylabel= {symm. diff.},
    width = \imgwidthhet,
    height = 2cm,
	title = {\fracking},
    cycle list name=yaf,
    scale only axis,
    x tick label style = {/pgf/number format/set thousands separator = {\,}},
    y tick label style = {/pgf/number format/set thousands separator = {\,}},
    scaled ticks = false,
    xtick = {10, 20, 30, 40, 50},
	every axis plot post/.append style= {line width=1pt},
    ]
\addplot table[x index = 0, y index = 1, header = true] {heterogeneous_experimental_data_fracking.txt};
\addplot table[x index = 0, y index = 2, header = true] {heterogeneous_experimental_data_fracking.txt};
\addplot table[x index = 0, y index = 4, header = true] {heterogeneous_experimental_data_fracking.txt};
\addplot table[x index = 0, y index = 3, header = true] {heterogeneous_experimental_data_fracking.txt};
\pgfplotsextra{\yafdrawaxis{5}{50}{2397}{5452}}
\end{axis}
\end{tikzpicture}

\begin{tikzpicture}
\begin{axis}[xlabel={budget $k$},ylabel= {symm. diff.},
    width = \imgwidthhet,
    height = 2cm,
	title = {\iphone},
    cycle list name=yaf,
    scale only axis,
    x tick label style = {/pgf/number format/set thousands separator = {\,}},
    y tick label style = {/pgf/number format/set thousands separator = {\,}},
    scaled ticks = false,
    xtick = {10, 20, 30, 40, 50},
	every axis plot post/.append style= {line width=1pt},
    ]
\addplot table[x index = 0, y index = 1, header = true] {heterogeneous_experimental_data_iphone_samsung.txt};
\addplot table[x index = 0, y index = 2, header = true] {heterogeneous_experimental_data_iphone_samsung.txt};
\addplot table[x index = 0, y index = 4, header = true] {heterogeneous_experimental_data_iphone_samsung.txt};
\addplot table[x index = 0, y index = 3, header = true] {heterogeneous_experimental_data_iphone_samsung.txt};
\pgfplotsextra{\yafdrawaxis{5}{50}{295}{682}}
\end{axis}
\end{tikzpicture}%
\begin{tikzpicture}
\begin{axis}[xlabel={budget $k$},ylabel= {symm. diff.},
    width = \imgwidthhet,
    height = 2cm,
	title = {\obamacare},
    cycle list name=yaf,
    scale only axis,
    x tick label style = {/pgf/number format/set thousands separator = {\,}},
    y tick label style = {/pgf/number format/set thousands separator = {\,}},
    scaled ticks = false,
    xtick = {10, 20, 30, 40, 50},
	every axis plot post/.append style= {line width=1pt},
    ]
\addplot table[x index = 0, y index = 1, header = true] {heterogeneous_experimental_data_obamacare.txt};
\addplot table[x index = 0, y index = 2, header = true] {heterogeneous_experimental_data_obamacare.txt};
\addplot table[x index = 0, y index = 4, header = true] {heterogeneous_experimental_data_obamacare.txt};
\addplot table[x index = 0, y index = 3, header = true] {heterogeneous_experimental_data_obamacare.txt};
\pgfplotsextra{\yafdrawaxis{5}{50}{1235}{2648}}
\end{axis}
\end{tikzpicture}%
\begin{tikzpicture}
\begin{axis}[xlabel={budget $k$},ylabel= {symm. diff.},
    width = \imgwidthhet,
    height = 2cm,
	every axis plot post/.append style= {line width=1pt},
    cycle list name=yaf,
    scale only axis,
	title = {\uselections},
    x tick label style = {/pgf/number format/set thousands separator = {\,}},
    y tick label style = {/pgf/number format/set thousands separator = {\,}},
    scaled ticks = false,
    xtick = {10, 20, 30, 40, 50},
    legend style={nodes={scale=0.75, transform shape}},
	legend entries = {\coveralg, \hedgealg, \commonalg, \greedyalg}
    ]
\addplot table[x index = 0, y index = 1, header = true] {heterogeneous_experimental_data_uselections.txt};
\addplot table[x index = 0, y index = 2, header = true] {heterogeneous_experimental_data_uselections.txt};
\addplot table[x index = 0, y index = 4, header = true] {heterogeneous_experimental_data_uselections.txt};
\addplot table[x index = 0, y index = 3, header = true] {heterogeneous_experimental_data_uselections.txt};
\pgfplotsextra{\yafdrawaxis{5}{50}{231}{1934}}
\end{axis}
\end{tikzpicture}
\caption{Expected symmetric difference $n - \diffind$ as a function of the budget $k$. Heterogeneous model. Low values are better.}
\label{fig:heterogeneous_score}
\end{figure}
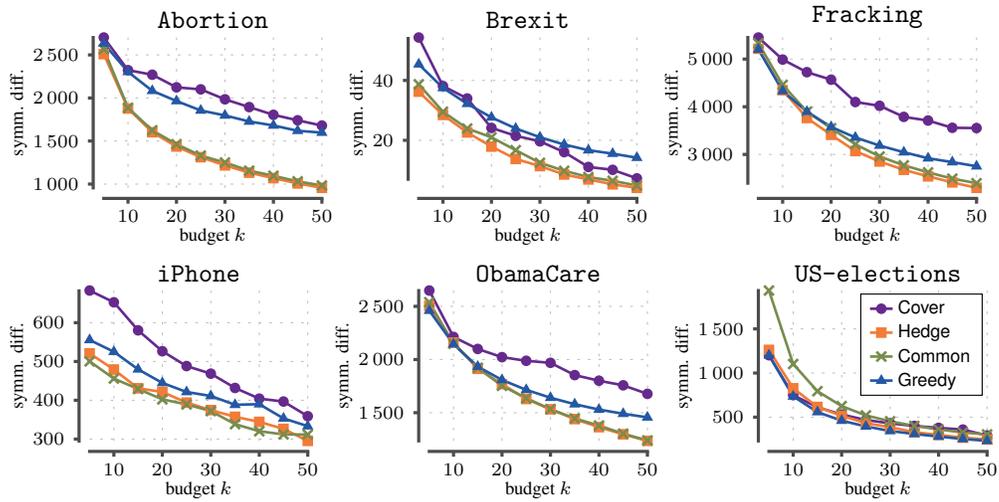
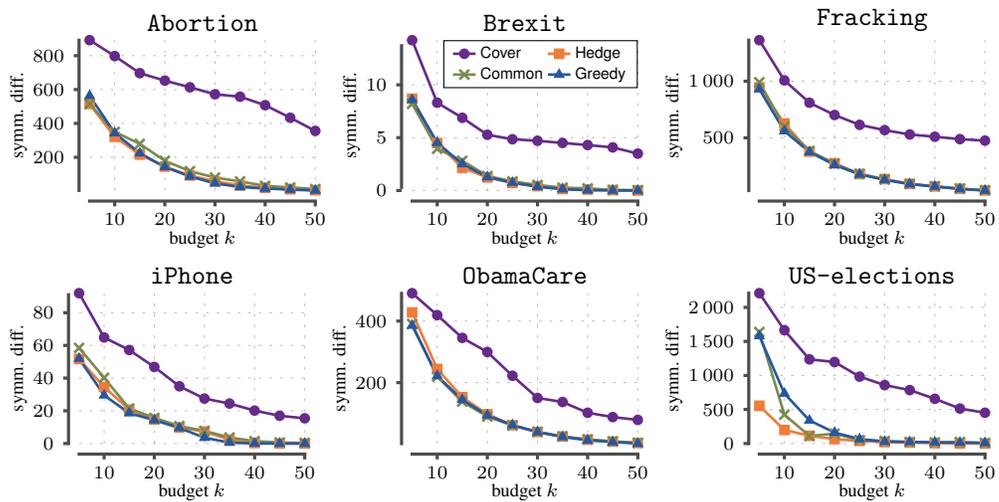
\begin{figure}[h]
\setlength{\imgwidth}{3cm}
\begin{tikzpicture}
\begin{axis}[xlabel={budget $k$},ylabel= {symm. diff.},
    width = \imgwidth,
    height = 2cm,
	title = {\abortion},
    cycle list name=yaf,
    scale only axis,
    x tick label style = {/pgf/number format/set thousands separator = {\,}},
    y tick label style = {/pgf/number format/set thousands separator = {\,}},
    scaled ticks = false,
    xtick = {10, 20, 30, 40, 50},
	every axis plot post/.append style= {line width=1pt},
    ]
\addplot table[x index = 0, y index = 1, header = true] {homogeneous_experimental_data_abortion.txt};
\addplot table[x index = 0, y index = 2, header = true] {homogeneous_experimental_data_abortion.txt};
\addplot table[x index = 0, y index = 4, header = true] {homogeneous_experimental_data_abortion.txt};
\addplot table[x index = 0, y index = 3, header = true] {homogeneous_experimental_data_abortion.txt};
\pgfplotsextra{\yafdrawaxis{5}{50}{4}{891}}
\end{axis}
\end{tikzpicture}%
\begin{tikzpicture}
\begin{axis}[xlabel={budget $k$},ylabel= {symm. diff.},
    width = \imgwidth,
    height = 2cm,
	title = {\brexit},
    cycle list name=yaf,
    scale only axis,
    x tick label style = {/pgf/number format/set thousands separator = {\,}},
    y tick label style = {/pgf/number format/set thousands separator = {\,}},
    scaled ticks = false,
    xtick = {10, 20, 30, 40, 50},
	every axis plot post/.append style= {line width=1pt},
	legend style={nodes={scale=0.65, transform shape}},
	legend entries = {\coveralg, \hedgealg, \commonalg, \greedyalg},
	legend columns = 2,
	every axis legend/.append style={at={(1,1)}, anchor=north east}
    ]
\addplot table[x index = 0, y index = 1, header = true] {homogeneous_experimental_data_brexit.txt};
\addplot table[x index = 0, y index = 2, header = true] {homogeneous_experimental_data_brexit.txt};
\addplot table[x index = 0, y index = 4, header = true] {homogeneous_experimental_data_brexit.txt};
\addplot table[x index = 0, y index = 3, header = true] {homogeneous_experimental_data_brexit.txt};
\pgfplotsextra{\yafdrawaxis{5}{50}{14}{0}}
\end{axis}
\end{tikzpicture}%
\begin{tikzpicture}
\begin{axis}[xlabel={budget $k$},ylabel= {symm. diff.},
    width = \imgwidth,
    height = 2cm,
	title = {\fracking},
    cycle list name=yaf,
    scale only axis,
    x tick label style = {/pgf/number format/set thousands separator = {\,}},
    y tick label style = {/pgf/number format/set thousands separator = {\,}},
    scaled ticks = false,
    xtick = {10, 20, 30, 40, 50},
	every axis plot post/.append style= {line width=1pt},
    ]
\addplot table[x index = 0, y index = 1, header = true] {homogeneous_experimental_data_fracking.txt};
\addplot table[x index = 0, y index = 2, header = true] {homogeneous_experimental_data_fracking.txt};
\addplot table[x index = 0, y index = 4, header = true] {homogeneous_experimental_data_fracking.txt};
\addplot table[x index = 0, y index = 3, header = true] {homogeneous_experimental_data_fracking.txt};
\pgfplotsextra{\yafdrawaxis{5}{50}{34}{1362}}
\end{axis}
\end{tikzpicture}
\begin{tikzpicture}
\begin{axis}[xlabel={budget $k$},ylabel= {symm. diff.},
    width = \imgwidth,
    height = 2cm,
	title = {\iphone},
    cycle list name=yaf,
    scale only axis,
    x tick label style = {/pgf/number format/set thousands separator = {\,}},
    y tick label style = {/pgf/number format/set thousands separator = {\,}},
    scaled ticks = false,
    xtick = {10, 20, 30, 40, 50},
	every axis plot post/.append style= {line width=1pt},
    ]
\addplot table[x index = 0, y index = 1, header = true] {homogeneous_experimental_data_iphone_samsung.txt};
\addplot table[x index = 0, y index = 2, header = true] {homogeneous_experimental_data_iphone_samsung.txt};
\addplot table[x index = 0, y index = 4, header = true] {homogeneous_experimental_data_iphone_samsung.txt};
\addplot table[x index = 0, y index = 3, header = true] {homogeneous_experimental_data_iphone_samsung.txt};
\pgfplotsextra{\yafdrawaxis{5}{50}{0}{91}}
\end{axis}
\end{tikzpicture}%
\begin{tikzpicture}
\begin{axis}[xlabel={budget $k$},ylabel= {symm. diff.},
    width = \imgwidth,
    height = 2cm,
	title = {\obamacare},
    cycle list name=yaf,
    scale only axis,
    x tick label style = {/pgf/number format/set thousands separator = {\,}},
    y tick label style = {/pgf/number format/set thousands separator = {\,}},
    scaled ticks = false,
    xtick = {10, 20, 30, 40, 50},
	every axis plot post/.append style= {line width=1pt},
    ]
\addplot table[x index = 0, y index = 1, header = true] {homogeneous_experimental_data_obamacare.txt};
\addplot table[x index = 0, y index = 2, header = true] {homogeneous_experimental_data_obamacare.txt};
\addplot table[x index = 0, y index = 4, header = true] {homogeneous_experimental_data_obamacare.txt};
\addplot table[x index = 0, y index = 3, header = true] {homogeneous_experimental_data_obamacare.txt};
\pgfplotsextra{\yafdrawaxis{5}{50}{2}{490}}
\end{axis}
\end{tikzpicture}%
\begin{tikzpicture}
\begin{axis}[xlabel={budget $k$},ylabel= {symm. diff.},
    width = \imgwidth,
    height = 2cm,
	every axis plot post/.append style= {line width=1pt},
    cycle list name=yaf,
    scale only axis,
	title = {\uselections},
    x tick label style = {/pgf/number format/set thousands separator = {\,}},
    y tick label style = {/pgf/number format/set thousands separator = {\,}},
    scaled ticks = false,
    xtick = {10, 20, 30, 40, 50},
    ]
\addplot table[x index = 0, y index = 1, header = true] {homogeneous_experimental_data_uselections.txt};
\addplot table[x index = 0, y index = 2, header = true] {homogeneous_experimental_data_uselections.txt};
\addplot table[x index = 0, y index = 4, header = true] {homogeneous_experimental_data_uselections.txt};
\addplot table[x index = 0, y index = 3, header = true] {homogeneous_experimental_data_uselections.txt};
\pgfplotsextra{\yafdrawaxis{5}{50}{0}{2208}}
\end{axis}
\end{tikzpicture}
\caption{Expected symmetric difference $n - \difftc$ as a function of the budget $k$. Correlated model. Low values are better.}
\label{fig:homogeneous_score}
\end{figure}
\begin{figure}
\pgfkeys{
  /pgf/number format/sci generic={mantissa sep={\times},exponent={10^{#1}}},
  }
\begin{center}
\begin{tikzpicture}
\begin{axis}[ylabel= {running time (s)}, xlabel = {number of edges, $\abs{E}$},
    width = 10cm,
    height = 5cm,
    cycle list name=yaf,
    scale only axis,
    tick scale binop=\times,
    x tick label style = {/pgf/number format/set thousands separator = {\,}},
    y tick label style = {/pgf/number format/set thousands separator = {\,}},
    xtick = {0, 500000, 1000000, 1500000},
	scaled ticks = false,
	xmin = 0,
	every axis plot post/.append style= {line width=1pt},
	bar width = 2.5pt,
	legend pos = {south east},
	legend entries = {\coveralg, \hedgealg, \commonalg}
    ]

\addplot+[only marks] table[x index = 9, y index = 5, header = true] {time_experimental_data_comparison.txt};
\addplot+[only marks] table[x index = 9, y index = 6, header = true] {time_experimental_data_comparison.txt};
\addplot+[only marks] table[x index = 9, y index = 8, header = true] {time_experimental_data_comparison.txt};

\addplot[no markers, yafcolor1] table[x index = 9, y = {create col/linear regression={x = edges, y = Hom-Cover}}, header = true] {time_experimental_data_comparison.txt};
\addplot[no markers, yafcolor2] table[x index = 9, y = {create col/linear regression={x = edges, y = Hom-Hedge}}, header = true] {time_experimental_data_comparison.txt};
\addplot[no markers, yafcolor3] table[x index = 9, y = {create col/linear regression={x = edges, y = Hom-Common}}, header = true] {time_experimental_data_comparison.txt};

\pgfplotsextra{\yafdrawaxis{0}{1511670}{0}{804}}
\end{axis}
\end{tikzpicture}%
\end{center}
\caption{Running time as a function of number of edges. Correlated model with $k = 20$.}
\label{fig:running_time}
\end{figure}
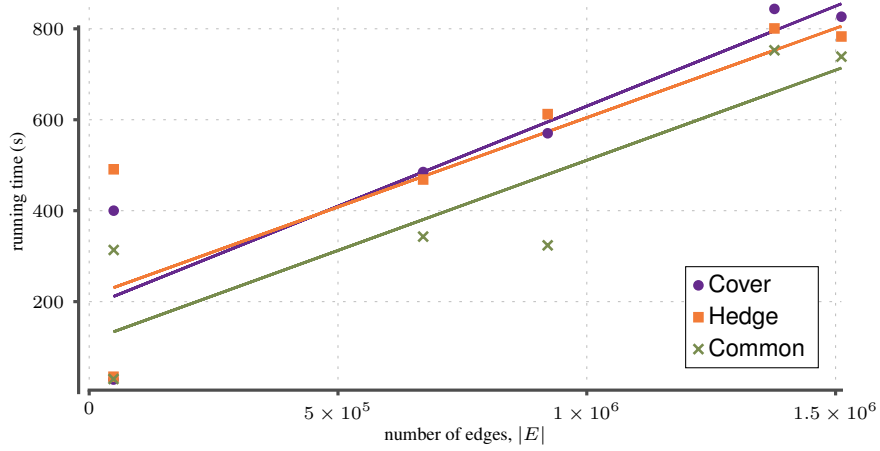

\begin{figure}
\setlength{\tabcolsep}{1pt}
\begin{tabular}{lll}
Side 1 & Side 2 & \hedgealg \\
\emph{Pro-Choice} & \emph{Pro-Life}\\
\includegraphics[width = 4.5cm]{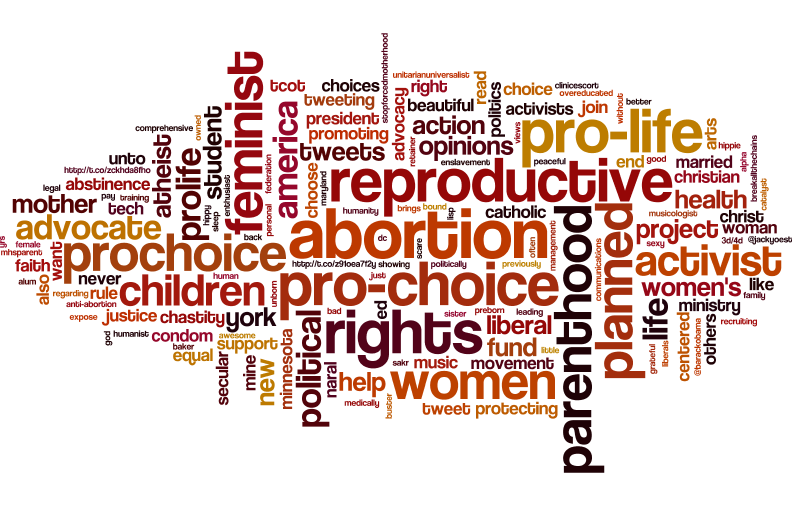} &
\includegraphics[width = 4.5cm]{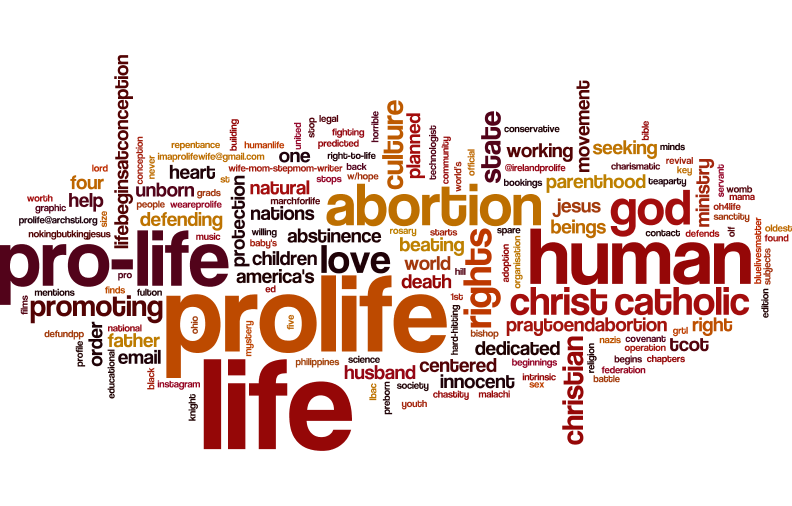} &
\includegraphics[width = 4.5cm]{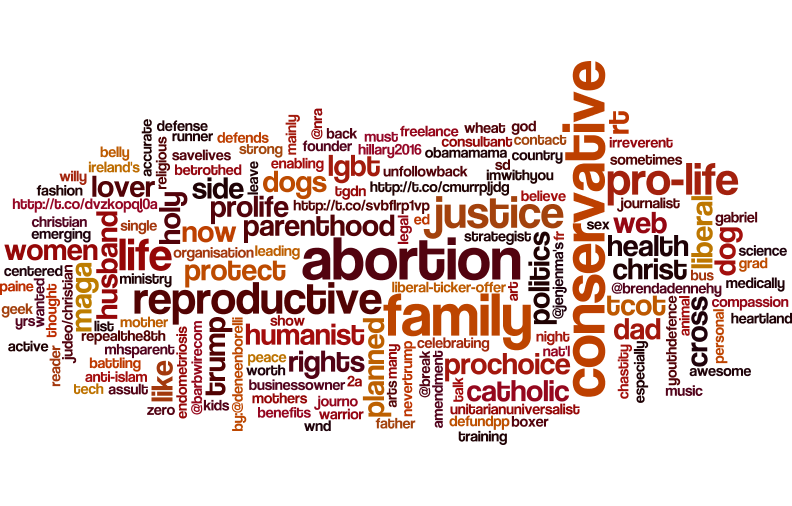} \\

\emph{Pro-Remain} & \emph{Pro-Leave} \\

\includegraphics[width = 4.5cm]{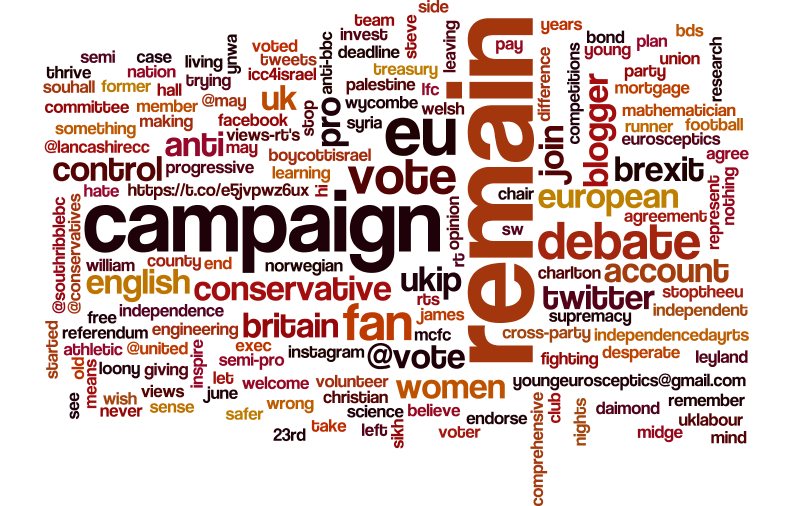} &
\includegraphics[width = 4.5cm]{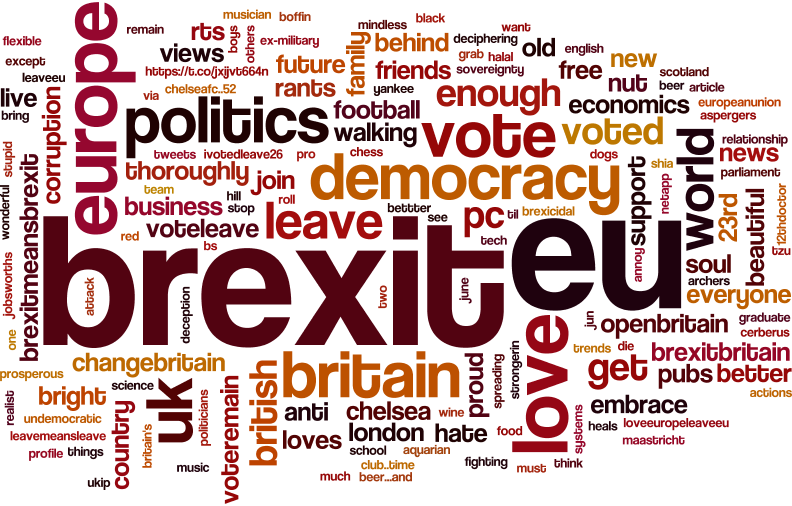} &
\includegraphics[width = 4.5cm]{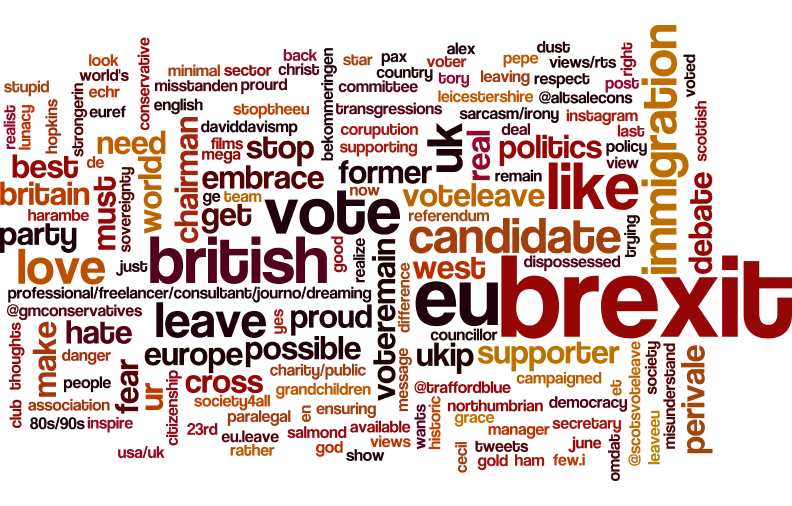} \\

\emph{Pro-Fracking} & \emph{Anti-Fracking} \\

\includegraphics[width = 4.5cm]{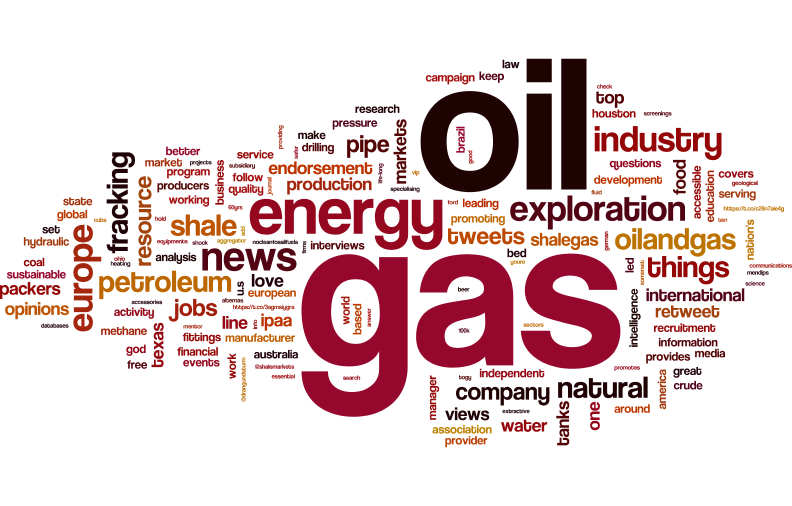} &
\includegraphics[width = 4.5cm]{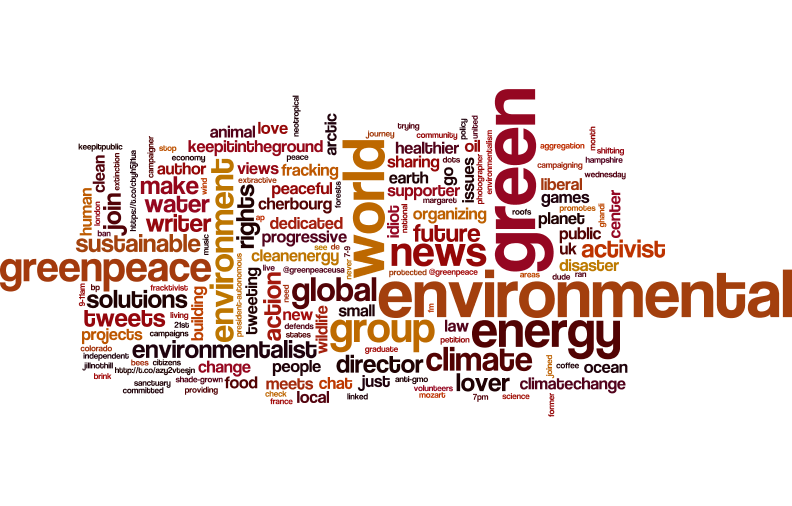} &
\includegraphics[width = 4.5cm]{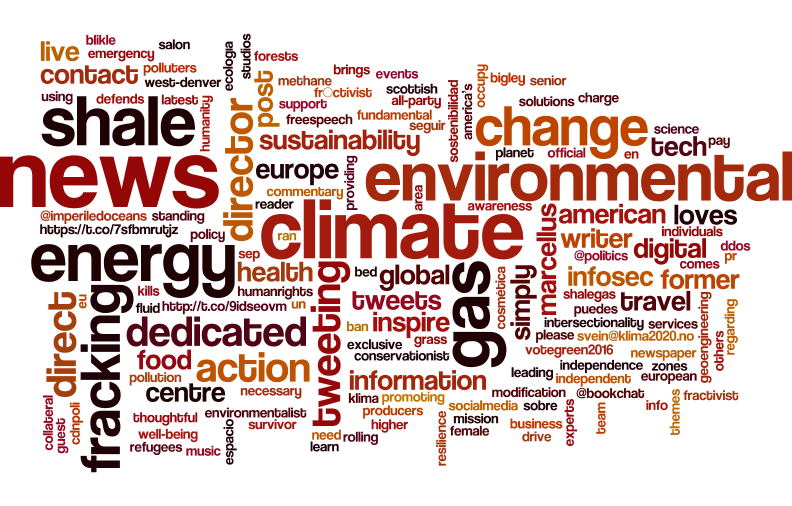} \\

\emph{Pro-iPhone} & \emph{Pro-Samsung} \\

\includegraphics[width = 4.5cm]{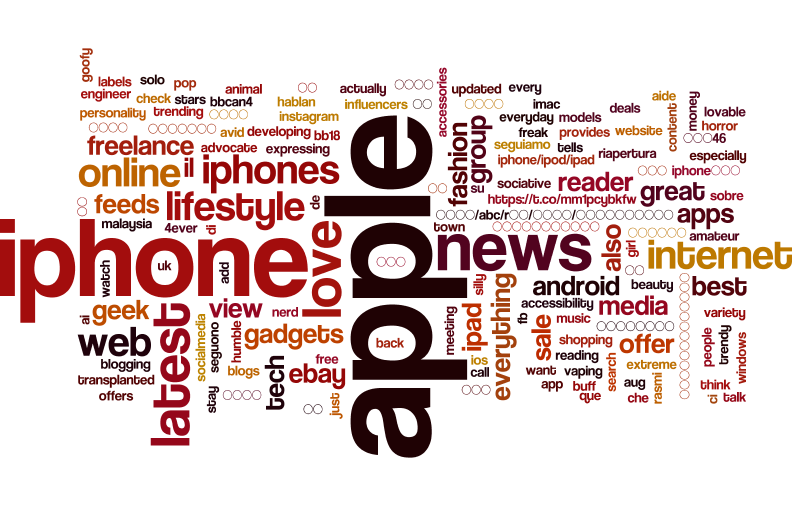} &
\includegraphics[width = 4.5cm]{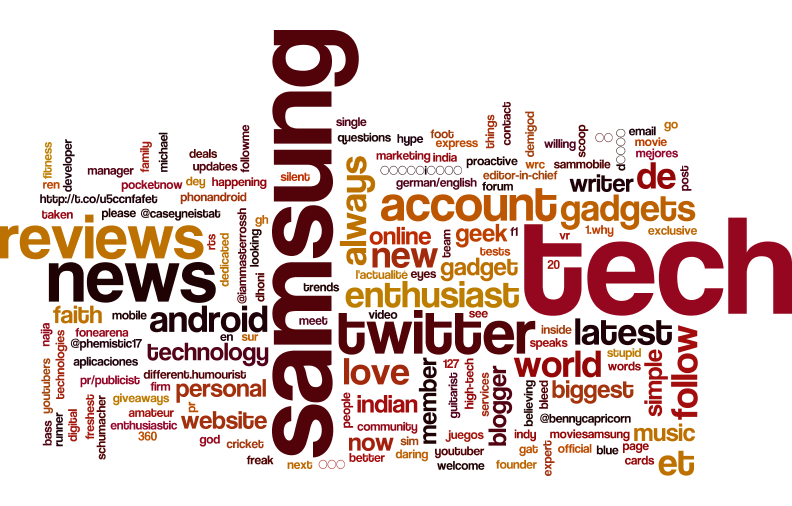} &
\includegraphics[width = 4.5cm]{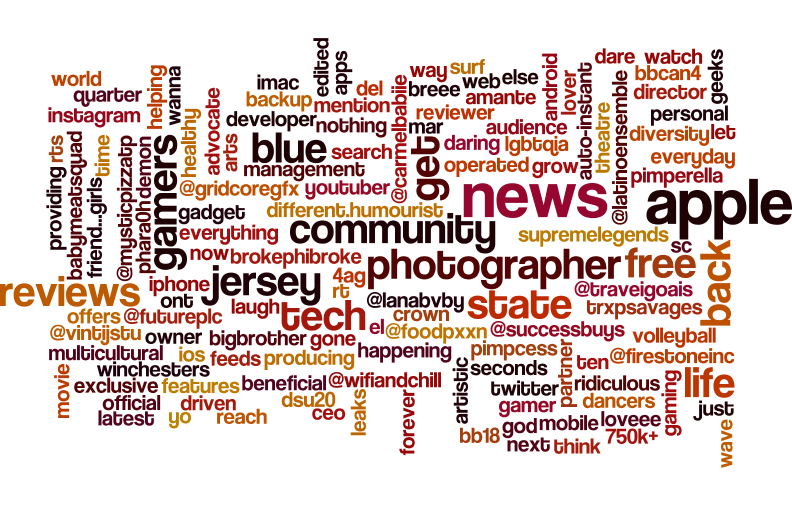} \\

\emph{Pro-Obamacare} & \emph{Anti-Obamacare} \\

\includegraphics[width = 4.5cm]{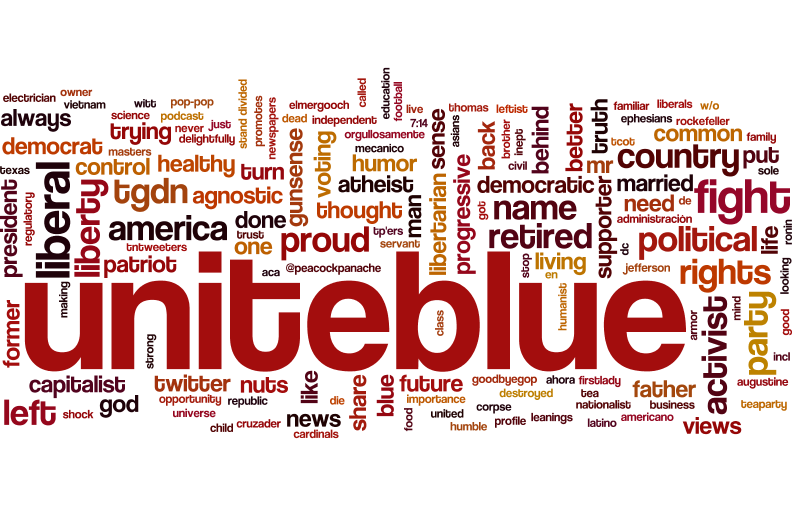} &
\includegraphics[width = 4.5cm]{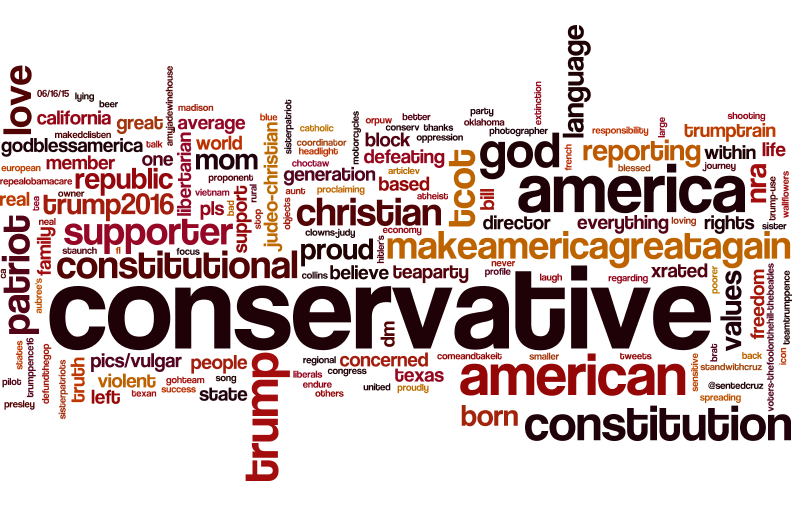} &
\includegraphics[width = 4.5cm]{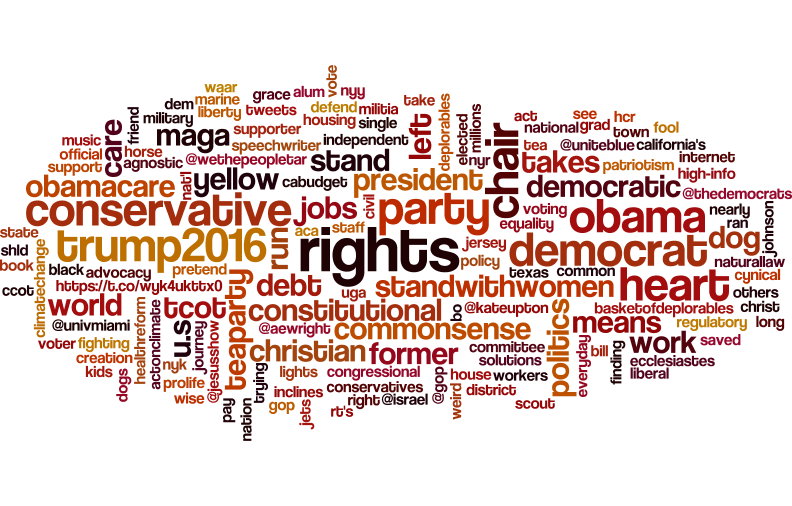} \\

\emph{Pro-Hillary} & \emph{Pro-Trump} \\

\includegraphics[width = 4.5cm]{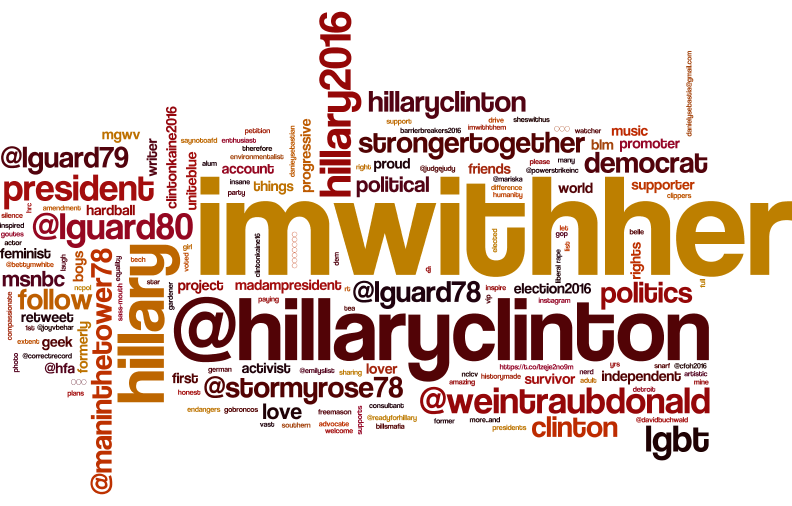} &
\includegraphics[width = 4.5cm]{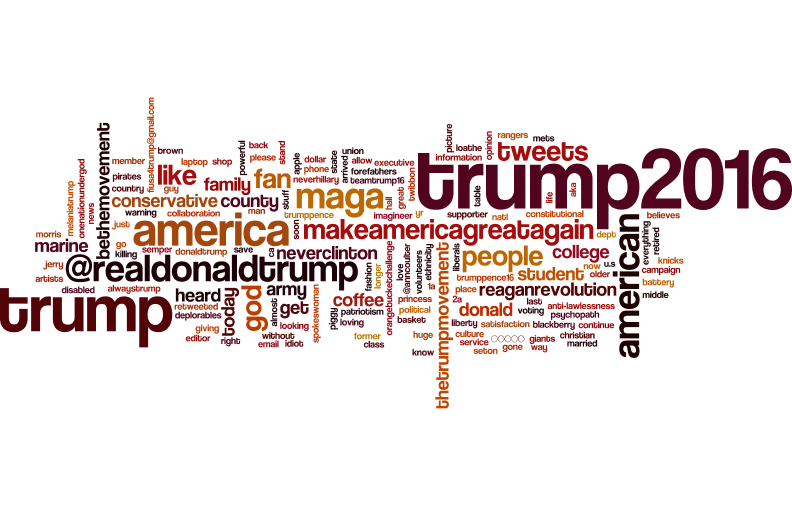} &
\includegraphics[width = 4.5cm]{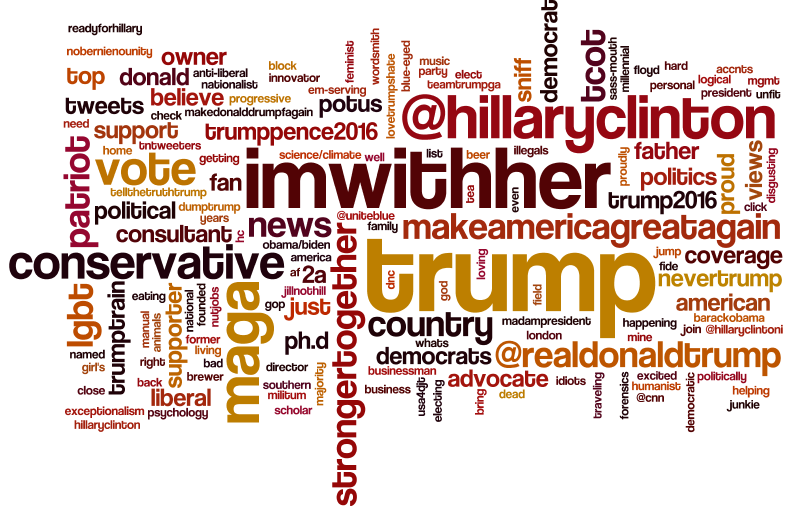} \\
\end{tabular}
\caption{Word clouds of the profiles for the initial seeds, and profiles selected by \hedgealg.}
\label{fig:clouds}
\end{figure}

\end{document}